\documentclass[12pt]{article}

\usepackage[letterpaper, margin=1in]{geometry}
\usepackage{graphicx}
\usepackage{natbib}
\usepackage{bm}
\usepackage{float}
\usepackage{siunitx}
\usepackage{multirow}
\usepackage{booktabs}
\usepackage{amsmath}
\usepackage{amssymb}
\usepackage{amsfonts}
\usepackage{amsthm}
\usepackage{caption}
\usepackage{hyperref}

\def\BibTeX{{\rm B\kern-.05em{\sc i\kern-.025em b}\kern-.08em
    T\kern-.1667em\lower.7ex\hbox{E}\kern-.125emX}}

\newtheorem{theorem}{Theorem}
\newtheorem{lemma}{Lemma}
\newtheorem{corollary}{Corollary}
\newtheorem{proposition}{Proposition}
\theoremstyle{remark}
\newtheorem{remark}{Remark}

\def\UnderWiggleTemp{\the\catcode`\@}
\catcode`\@=11
\ifx\UnderWiggle@Loaded\relax
  \catcode`\@=\UnderWiggleTemp
  \endinput \else \let\UnderWiggle@Loaded=\relax \fi
\newbox\U@BoxA
\newbox\U@BoxB
\newdimen\U@DimenA
\def\U@DoUnderWiggle{
  \offinterlineskip
  \vtop{
    \hbox{\vbox{\copy0}}
    \vskip 1.2pt
    \vbox to 0.4pt{
      \hbox to\wd0{\hss\char'176\hss}
      \vskip0pt minus 1fil
    }
    \vskip 0.4pt
  }
}
\def\UnderWiggle#1{{%
  \ifmmode
    \mathchoice
      {\setbox0=\hbox{$\displaystyle #1$}\U@DoUnderWiggle}
      {\setbox0=\hbox{$\textstyle #1$}\U@DoUnderWiggle}
      {\setbox0=\hbox{$\scriptstyle #1$}\U@DoUnderWiggle}
      {\setbox0=\hbox{$\scriptscriptstyle #1$}\U@DoUnderWiggle}
  \else
    \setbox0=\hbox{#1}\U@DoUnderWiggle
  \fi
}}
\catcode`\@=\UnderWiggleTemp

\newcommand{\be}{\begin{eqnarray}}
\newcommand{\uw}{\UnderWiggle}
\newcommand{\ee}{\end{eqnarray}}
\newcommand{\ba}{\begin{eqnarray*}}
\newcommand{\ea}{\end{eqnarray*}}

\title{Composite Wavelet Matrix--Based Transforms and Applications}

\author{
Radhika R. Kulkarni \and Brani Vidakovic\\
{\large Department of Statistics, Texas A\&M University, College Station, TX, USA}
}

\date{\today}

\begin{document}

\maketitle

\begin{abstract}
Orthogonal wavelet transforms are a cornerstone of modern signal and image denoising because they
combine multiscale representation, energy preservation, and perfect reconstruction. In this paper,
we show that these advantages can be retained and enhanced by moving beyond classical
single–basis wavelet filterbanks to a broader class of \emph{composite wavelet–like matrices}.
By combining orthogonal wavelet matrices through products, Kronecker products, and block–diagonal
constructions, we obtain new unitary transforms that generally fall outside the strict wavelet
filterbank class, yet remain fully invertible and numerically stable.

The central finding is that such composite transforms often induce stronger concentration of signal energy
into fewer coefficients than conventional wavelets. This increased sparsity, quantified using
Lorenz curve diagnostics, directly translates into improved denoising under identical
thresholding rules. Extensive simulations on Donoho--Johnstone benchmark signals, complex–valued
unitary examples, and adaptive block constructions demonstrate consistent reductions in
mean–squared error relative to single–basis transforms. Applications to atmospheric turbulence measurements and image denoising of the Barbara benchmark further
confirm that composite transforms better preserve salient structures while suppressing noise.

From an engineering perspective, the proposed framework is attractive because it requires no modification
to the standard transform–shrink–inverse pipeline and preserves orthogonality and perfect
reconstruction by construction. The results suggest a shift in viewpoint: wavelet transforms can be
treated as modular orthogonal operators that may be algebraically combined to better match signal
structure. Composite wavelet–like matrices therefore provide a practical and flexible extension of
classical wavelet methods for denoising, compression, and multiscale analysis in one– and
higher–dimensional settings.
\medskip

\end{abstract}

\section{Introduction}

Wavelet-based transforms have become a cornerstone of modern signal processing,
statistical estimation, and data compression because of their ability to represent signals
efficiently across multiple scales.
In particular, orthogonal discrete wavelet transforms (DWTs) provide sparse
representations for a wide range of signals and enable stable denoising through simple
coefficient-wise shrinkage rules.
When combined with thresholding, orthogonal wavelet bases yield estimators that are
computationally efficient, invertible, and near-optimal for many classes of functions.

Despite these successes, a fundamental limitation of classical wavelet methodology
is its reliance on a \emph{single} wavelet basis.
The performance of wavelet denoising is strongly dependent on the choice of basis, a well-established fact in the denoising literature~\citep{Sahoo2024,Srivastava2016}, 
and no single orthogonal wavelet family is uniformly optimal across signals with
heterogeneous structure.
In practice, signals often exhibit a mixture of smooth components,
localized spikes, oscillatory textures, or intermittent bursts, each of which may be
best represented by different wavelet atoms.
Consequently, committing to a single wavelet basis can lead to oversmoothing in some
regions and insufficient noise suppression in others.

From an operator-theoretic perspective, the essential property enabling wavelet
denoising is not the filterbank construction itself, but rather \emph{orthogonality}
(or unitarity in the complex setting).
Orthogonality guarantees energy preservation, numerical stability, and perfect
reconstruction, which are precisely the properties required for threshold-based
shrinkage and inverse transformation.
This observation motivates a broader question: rather than restricting attention to
classical wavelet filterbanks, can one enlarge the class of admissible transforms by
considering general orthogonal operators constructed from wavelet matrices?

In this paper, we answer this question affirmatively by introducing
\emph{composite wavelet-like matrices}.
These are orthogonal (or unitary) matrices formed by algebraic operations on
orthogonal wavelet matrices, including matrix products, Kronecker products,
block-diagonal assemblies, and similarity transforms.
Although such composites may not correspond to standard wavelet filterbanks,
they retain the crucial properties of orthogonality and invertibility.
As linear operators, they can therefore be used directly in multiscale signal
representations, coefficient shrinkage, and reconstruction. Throughout, these operators act on finite-length discrete signals and images, i.e., vectors in $\mathbb{R}^N$ with $N=2^J$,  $J=\log_2(N)$ (and their two-dimensional analogs through the Kronecker construction).

The central contribution of this work is to demonstrate that composite wavelet-like
matrices can yield systematically improved sparsity and denoising performance
relative to single-wavelet bases.
To isolate the effect of the transform itself, we fix the thresholding rule across all
experiments and compare performance solely on the basis of the induced coefficient
representations.
Through extensive simulations using the Donoho--Johnstone benchmark signals,
we show that composite transforms often produce more disbalanced energy distributions
and lower mean-squared error (MSE) after shrinkage than any individual wavelet basis.
These improvements persist across different signal types and noise levels.

Beyond one-dimensional denoising, composite wavelet-like matrices offer a flexible
framework for higher-dimensional data, adaptive representations, and structured
operator design.
Kronecker constructions naturally extend to multidimensional signals and images,
while block-diagonal assemblies allow different wavelet structures to be applied
locally to heterogeneous signal segments.
Moreover, because the framework is purely algebraic, it aligns naturally with
tensor-product representations arising in modern computational architectures,
including quantum information processing.

The remainder of the paper is organized as follows.
Section~\ref{sec:II} introduces orthogonal wavelet matrices from a linear-operator viewpoint,
emphasizing their role as building blocks for more general constructions.
Section~\ref{sec:III} describes several algebraic operations that preserve orthogonality and
give rise to composite wavelet-like matrices.
Section~\ref{sec:IV} presents a comprehensive simulation study comparing composite and
single-basis transforms under identical shrinkage rules.
Section~\ref{sec:V} illustrates the advantages of composite transforms on real-world data.
Section~\ref{sec:VI} concludes with discussion and directions for future research. Appendix~\ref{sec:A} proves Proposition~\ref{prop:prod_not_wavelet},  Appendix~\ref{sec:B}  gives additional turbulence results, Appendix~\ref{sec:C}  proves Corollary~\ref{cor:universal}, Appendix~\ref{sec:D}  reports the noise-level and filter-pair sweeps, and Appendix~\ref{sec:E}  proves Proposition~\ref{prop:nonuniversality}.
%%%%%%%%%%%%%%%%%%%%%%%%%%%%%
% ORTHOGONAL WAVELET MATRICES
%%%%%%%%%%%%%%%%%%%%%%%%%%%%%
\section{Orthogonal Wavelet Matrices}
\label{sec:II}
\noindent
This section establishes the matrix-based formulation that underpins all composite constructions developed later.

Orthogonal wavelet matrices provide a concrete linear--algebraic representation of the
discrete wavelet transform (DWT). Instead of viewing the DWT exclusively as a recursive
filtering and downsampling algorithm, the matrix formulation expresses the transform as
an orthogonal change of basis in $\mathbb{R}^N$. This viewpoint makes the structural properties of the transform explicit---orthogonality, energy preservation, and perfect reconstruction---and  enables systematic manipulation and composition of wavelet transforms using standard matrix operations. These properties are central to the
composite constructions developed later in the paper.

Let the input signal $\uw y$ have length $N=2^J$, where $J=\log_2(N)$. An $L$--level
orthogonal wavelet transform maps $\uw y$ to a vector of scaling and wavelet coefficients
$\uw d$ via a matrix multiplication

\[
\uw d = \bm W \uw y,
\]
where $\bm W$ is an $N \times N$ orthogonal matrix. The transform produces wavelet detail
coefficients at levels $J-1, J-2, \dots, J-L$, together with scaling coefficients at level
$J-L$, where $L \le J$ is the number of decomposition levels.

\subsection{Single--level wavelet matrices}

Let $\uw h=\{h_s,\, s\in\bm Z\}$ be the low--pass wavelet filter, and let $n$ be a fixed
integer shift parameter. For each level $k=1,2,\dots$, define a matrix $H_k$ of size
$(2^{J-k} \times 2^{J-k+1})$ whose $(i,j)$ entry is
\be
\label{eq:modulo}
h_s, \qquad s = n + j - 2i \ \mbox{\it modulo} \ 2^{J-k+1}.
\ee
The modulo operation implements periodic boundary conditions.
Successive rows of $H_k$ are circular shifts by two samples,
reflecting filtering followed by decimation by two and permitting
an efficient matrix implementation.

By analogy, define the matrix $G_k$, of the same size as $H_k$, using the high--pass filter $\uw g$, where $g_i =
(-1)^i h_{n+1-i}$. The choice of $n$ determines the relative shift of the wavelet and
scaling functions on the time grid; for filters from the Daubechies family, a standard
choice is to take $n$ equal to the number of vanishing moments \citep[see][]{Daubechies1992ten,Mallat2009wavelet}.

The stacked matrix

\[
\left[ \begin{array}{c} H_k \\ G_k \end{array} \right]
\]
acts as a change of basis in $\mathbb{R}^{2^{J-k+1}}$ and is therefore unitary. Denote $^{\top}$ as the transpose. It follows
that

\[
I_{2^{J-k+1}} = [H_k^{\top} \; G_k^{\top}] \left[ \begin{array}{c} H_k \\ G_k \end{array} \right]
= H_k^{\top} H_k + G_k^{\top} G_k.
\]
and

\[
\left[ \begin{array}{c} H_k \\ G_k \end{array} \right]
[H_k^{\top} \; G_k^{\top}]
=
\left[ \begin{array}{cc}
H_k H_k^{\top} & H_k G_k^{\top} \\
G_k H_k^{\top} & G_k G_k^{\top}
\end{array} \right]
= I.
\]
Consequently,
\[
H_k H_k^{\top} = I, \qquad G_k G_k^{\top} = I, \qquad
H_k G_k^{\top} = G_k H_k^{\top} = 0,
\]
and
\[
H_k^{\top} H_k + G_k^{\top} G_k = I.
\]
These identities express orthogonality and perfect reconstruction at a single
decomposition level.

\subsection{Multilevel wavelet matrices}

The full orthogonal wavelet transform is obtained by cascading these single--level
orthogonal changes of basis across scales. Writing $\bm W_L$ for the $L$--level
wavelet $N\times N$ matrix, one obtains
\ba
\bm W_1 &=& \left[ \begin{array}{c} H_1 \\ G_1 \end{array} \right], \\
\bm W_2 &=& \left[ \begin{array}{c}
\left[ \begin{array}{c} H_2 \\ G_2 \end{array} \right] H_1 \\
G_1
\end{array} \right], \\
\bm W_3 &=& \left[ \begin{array}{c}
\left[ \begin{array}{c}
\left[ \begin{array}{c} H_3 \\ G_3 \end{array} \right] H_2 \\
G_2
\end{array} \right] H_1 \\
G_1
\end{array} \right], \dots
\ea
Proceeding in this way yields an orthogonal matrix $\bm W_L$ corresponding to an $L$--level
wavelet decomposition, producing detail coefficients at levels $J-1$ through $J-L$ and
scaling coefficients at level $J-L$. The orthogonality of $\bm W_L$ follows directly from the orthogonality of the individual blocks.

At this point, the filterbank interpretation becomes secondary; what matters for the
developments in subsequent sections is that $\bm W_L$ is an explicitly defined orthogonal
operator on $\mathbb{R}^N$.

Two consequences of orthogonality will be used throughout. First, $\bm W_L$ preserves energy: for any $\uw y \in \mathbb{R}^N$,
\[
\|\bm W_L \uw y\|^2
= \ \uw y^{\top} \bm W_L^{\top} \bm W_L \, \uw y
= \uw y^{\top} \uw y
= \|\uw y\|^2 .
\]
\noindent
Second, $\bm W_L$ is perfectly conditioned. That is, all its singular values equal
one, so $\kappa_2(\bm W_L) = 1$ and the transform neither amplifies nor
attenuates perturbations in the $\ell_2$ norm. The same two properties
hold for every orthogonal (or unitary) operator constructed in
Section~\ref{sec:III}, including the composite transforms, since they are
built to preserve orthogonality.

\subsection{Epsilon--decimated wavelet matrices}

We also consider the family of $\varepsilon$--decimated wavelet transforms
$\bm W_{\varepsilon}$. In an orthogonal DWT, each decomposition level involves a choice
between retaining even--indexed or odd--indexed coefficients after decimation by two.
Let

\[
\varepsilon = (\varepsilon_1, \varepsilon_2, \dots, \varepsilon_L),
\qquad \varepsilon_k \in \{0,1\}.
\]
Here, $\varepsilon_k = 0$ corresponds to retaining even--indexed coefficients at level
$k$, while $\varepsilon_k = 1$ similarly corresponds to retaining odd--indexed coefficients.
Operationally, this is implemented by modifying the index $s$ in
\eqref{eq:modulo} at each step $k$ by adding $\varepsilon_k$.

Each choice of $\varepsilon$ defines a distinct orthogonal wavelet matrix
$\bm W_{\varepsilon}$. Collectively, the $2^L$ such matrices represent all possible
phase--shifted versions of the DWT and form the basis of the stationary (non--decimated)
wavelet transform, as described by \citet{NasonSilverman1995}. In the present work, these
$\varepsilon$--decimated transforms are viewed as a natural library of orthogonal wavelet
matrices that can be combined through algebraic operations to form composite transforms. Other orthogonal libraries, such as wavelet-packet bases~\citep{CoifmanWickerhauser1992}, could equally serve as building blocks.

\subsection{Implementation and extensions}

Matrix implementations of orthogonal wavelet transforms
(\texttt{Wavmat.m|py}) and $\varepsilon$--decimated transforms
(\texttt{Wavmateps.m|py}) are available on GitHub at
\footnote{\url{https://github.com/rrkulkarni108/CompositeWavelets}}.
The explicit matrix representation facilitates experimentation with alternative
constructions and provides a direct computational bridge to the composite operators
studied later in the paper.

\medskip
\noindent
{\bf Comment.}
Wavelet matrices can also be constructed using Givens rotations
\citep{Givens1958}, which build orthogonal transformations through successive planar rotations. Although not pursued here, this approach offers another perspective on wavelet matrices as elements of the orthogonal group.

\medskip
\noindent
One may also consider unitary matrices constructed from complex-valued wavelet filters rather than purely real ones. Complex wavelet transforms preserve energy and perfect reconstruction while offering improved symmetry and phase information. An example involving the product $\bm W_1 \bm W_2$, with the complex-valued $\bm W_1$, is presented in Section~\ref{sec:3.1.1}.

\medskip
\noindent
Having established orthogonal wavelet transforms as explicit matrix operators, we now describe algebraic operations that preserve orthogonality and enable the construction of composite wavelet--like matrices.

\section{Combinations of Wavelet Matrices Preserving Orthogonality}
\label{sec:III}
% Bridge to Section 3
\noindent
Building on the explicit matrix formulation of orthogonal wavelet transforms introduced above, we now examine how such matrices can be combined while retaining orthogonality or unitarity.

\label{sec:sec3}

Section~\ref{sec:II} represents an $L$--level DWT on a signal of length $N=2^{J}$ as a multiplication by an $N\times N$ orthogonal (or unitary) matrix. This linear--operator viewpoint is the main device of the paper. Once a transform is expressed as a matrix, the design space expands: one can combine known wavelet matrices using standard matrix operations while retaining orthogonality (or unitarity). This is important for engineering applications because orthogonality guarantees energy preservation, numerical stability, and perfect reconstruction, so the usual pipeline (transform, shrink, inverse) remains valid even when the resulting operator is no longer a classical two--channel wavelet filterbank.

The most natural ways of combining wavelet matrices while preserving orthogonality or unitarity are the following:
taking matrix products, forming Kronecker products, constructing diagonal block matrices, and applying similarity transforms. 
In general, however, the resulting operator is \emph{not} a wavelet matrix in the strict filterbank sense --- i.e., it cannot be refactorized into a single two–channel polyphase form with the usual normalization and perfect reconstruction constraints. For products this is established in Proposition~\ref{prop:prod_not_wavelet} (proof in Appendix~\ref{sec:A}); the Kronecker and block–diagonal constructions fall outside the wavelet class by construction, as detailed in Sections~\ref{sec:sec3.2} and \ref{sec:sec3.3} respectively.
The main message is that one can and should consider a broader class of orthogonal ``wavelet--like'' operators when the goal is sparsity and denoising performance, rather than strict adherence to a single filterbank representation.

\subsection{Product of Wavelet Matrices}
\label{sec:sec3.1}

The product of two wavelet matrices preserves orthogonality but generally leaves the wavelet class.
Unless the composite transform can be refactorized back into the polyphase form of a wavelet filterbank
(with the normalization conditions satisfied), it is not itself a ``wavelet matrix'' in the strict sense.

Let

\[
\bm W = \bm W_1 \cdot \bm W_2
\]
for wavelet matrices $\bm W_1$ and $\bm W_2$ of the same size, where the subscripts now label two distinct wavelet matrices rather than a number of levels.
Then $\bm W$ is orthogonal/unitary, but in general need not be a wavelet matrix in the strict filterbank sense.

More generally, the product of multiple wavelet matrices can be written as

\[
\bm W = \prod_{i=1}^k \bm W_i
= \bm{W_1 \cdot W_2} \hspace{0.5em} \bm \cdots \hspace{0.5em} \bm{W_k},
\]
where each $\bm W_i$ is an orthogonal or unitary wavelet matrix of the same size.
This generalized product preserves orthogonality and energy, while offering greater flexibility to capture a variety of signal features, potentially enhancing sparsity and improving tasks such as compression or denoising. The sequential application of multiple transforms can provide a richer multi--scale representation than a single wavelet. This cascading of filtering stages is conceptually related to Mallat's scattering transform~\citep{BrunaMallat2013}, which also composes successive filters; the key difference is that the scattering transform applies a nonlinear modulus between stages, whereas our construction composes linear orthogonal operators and remains invertible.

As mentioned above, the resulting transform usually does not correspond to a single--stage classical
wavelet filter. The effective scaling function may be more complex or multiwavelet in nature, with the
standard normalization conditions $\sum h_n=\sqrt{2}$ and $\sum h_n^2=1$ generally not satisfied.
Specifically for the product of two wavelet matrices, we explored a possibility to collapse the composite
transform into a single filter pair by constructing a composite low--pass filter

\[
h = h_{w1} * h_{w2}^{\uparrow 2},
\]
where ${\uparrow \!\!2}$ is filter upsampling, and taking $g$ as its quadrature mirror filter. In the case where standard compactly supported orthonormal wavelet
families are used as product components, the half-band condition  required of a single orthonormal two-channel wavelet low-pass filter is not satisfied. Consequently, while $\bm W=\bm W_1 \bm W_2$ is an orthogonal and fully invertible transform, its atoms
cannot be generated by classical wavelet machinery; they must be read directly from the columns of
$\bm W^{\top}$ and form a broader class of orthonormal composite framelet elements that do not arise
from a single two--channel filter bank. Further discussion is given in Appendix~\ref{sec:appendix}.

\medskip

To make the above point explicit, we record the following proposition.

\medskip
\noindent
\begin{proposition}
\label{prop:prod_not_wavelet}
Let $\bm W_1$ and $\bm W_2$ be $N\times N$ orthogonal wavelet matrices constructed as in
Section~\ref{sec:II} from compactly supported orthonormal wavelet filters. Then $\bm W=\bm W_1 \bm W_2$  cannot be represented as a single two--channel perfect reconstruction wavelet filterbank;
that is, $\bm W$ is not a wavelet matrix in the strict polyphase sense.
\end{proposition}
\medskip
\noindent
A sketch of the proof, based on the collapse attempt $h=h_{w1}*h_{w2}^{\uparrow 2}$ and the failure of the resulting filter to satisfy the half-band (quadrature mirror) condition for standard compactly supported orthonormal families, is given in
Appendix~\ref{sec:A}. 
\medskip

\noindent 
So far, the orthogonality of $\bm W$ has been used to establish energy preservation, numerical stability, and perfect reconstruction, which are all properties of the transform and reconstructed signal. For denoising, orthogonality has a further consequence that we now make explicit, and which concerns the \emph{noise} rather than the signal: it leaves the noise i.i.d. Gaussian, so the universal threshold still applies. This holds for every composite operator of this section despite, as Proposition~\ref{prop:prod_not_wavelet} states, such operators not needing to be wavelet matrices. 

For denoising, write the noisy observation as
$\uw y = \uw s + \uw\varepsilon$, so that the transform coefficients
\[
\uw d = \bm W \uw y
      = \underbrace{\bm W \uw s}_{=\,\bm\theta} + \bm W \uw\varepsilon
\]
split into the clean--signal coefficients $\bm\theta = \bm W \uw s$ and a
noise term $\bm W \uw\varepsilon$. The following lemma characterizes the
latter.

\begin{lemma}[Noise invariance under orthogonal transforms]
\label{lem:noise}

\noindent
Let $\varepsilon \sim \mathcal{N}(0,\sigma^2 I_N)$ and let $\bm W$ be any
$N \times N$ orthogonal matrix. Then
$\bm W \varepsilon \sim \mathcal{N}(0,\sigma^2 I_N)$.
\end{lemma}

\begin{proof}
Any linear combination of a Gaussian vector is Gaussian, so $\bm W \varepsilon$ is Gaussian. It has mean $\bm W\,\mathbb{E}[\varepsilon] = 0$, and its
covariance is
\[
\mathbb{E}\big[(\bm W\varepsilon)(\bm W\varepsilon)^{\top}\big]
 = \bm W\,(\sigma^2 I_N)\,\bm W^{\top}
 = \sigma^2\, \bm W \bm W^{\top}
 = \sigma^2 I_N,
\]
with the last equality by orthogonality. Therefore, $\bm W\varepsilon$ and $\varepsilon$ have the same distribution.
\end{proof}

\medskip
\noindent
This then leads us to ask: beyond preserving the noise model, why might we be interested in such a matrix as $\bm W$?   
Depending on the number of levels of decomposition, $\bm W$ can
outperform a single--basis wavelet matrix in several scenarios. The key reason is
its ability to enhance \textit{disbalance} among wavelet coefficients, that is, to concentrate the energy of the
signal into just a few coefficients, more effectively than single--basis wavelet transforms. We can see
this visually using Lorenz curves \citep{lorenz1905}.

\noindent
The Lorenz curve provides a compact and intuitive way to visualize how energy is
distributed among wavelet coefficients. Let $d=(d_1,\dots,d_N)$ denote the wavelet
transform of a signal, and define the coefficient energies as

\[
x_i = |d_i|^2 , \qquad i=1,\dots,N,
\]
with $x_i \ge 0$. After sorting the energies in nondecreasing order,
$x_{(1)} \le x_{(2)} \le \cdots \le x_{(N)}$, the Lorenz curve is defined by
\begin{eqnarray}
L_N(p)
=
\frac{\sum_{i=1}^{\lfloor Np \rfloor} x_{(i)}}{\sum_{i=1}^{N} x_i},
\qquad 0 \le p \le 1.
\label{eq:lorenz}
\end{eqnarray}
The Lorenz curve plots the cumulative fraction of total signal energy captured by
the smallest fraction $p$ of wavelet coefficients. If the energy is evenly spread 
across the coefficients, then $L_N(p)\approx p$ and the curve lies close to the
$45$--degree line, indicating little or no compressibility. In contrast, when the
energy is concentrated in a relatively small number of coefficients, the smallest
$\lfloor Np\rfloor$ coefficients contribute very little energy for most values of
$p$, so that $L_N(p)\ll p$ over a wide range. The resulting curve (see explanatory graphic below) lies far below the
diagonal and rises sharply only near $p=1$.
According to the criterion of \citet{goel1995}, such coefficient disbalance is
preferable for signal representation and shrinkage-based denoising.

While the Lorenz curve provides a visual description of this
concentration, it is often convenient to summarize it numerically.
We therefore define the top-$k$ energy concentration by
\begin{equation}
C_k(d)
=
\frac{\displaystyle\sum_{j=1}^{k}|d|_{[j]}^2}
{\displaystyle\sum_{j=1}^{N}|d_j|^2},
\end{equation}
where $d=(d_1,\ldots,d_N)^\top$ is a nonzero coefficient vector and
$|d|_{[1]}\ge\cdots\ge |d|_{[N]}$
are the ordered coefficient magnitudes.
Larger values of $C_k(d)$ indicate greater $k$-term
compressibility. Moreover,
\begin{equation}
\label{eq:Ck}
C_k(d)
=
1-L_N\!\left(1-\frac{k}{N}\right),
\end{equation}

so $C_k$ and the Lorenz curve contain the same information. For denoising, what matters is not only how much energy the $k$ largest
coefficients retain but how much is lost when the remaining $N-k$ are
set to zero, since thresholding does exactly that. This loss is the
best-$k$-term approximation error of nonlinear approximation theory
\citep{DeVore1998,Mallat2009wavelet},
\begin{equation}
\varepsilon_k^2(d)
:=\min_{\|u\|_0\le k}\|d-u\|_2^2
=\sum_{j>k}|d|_{[j]}^2,
\qquad k=0,\dots,N,
\label{eq:best_k}
\end{equation}
where the minimum over all vectors with at most $k$ nonzero entries is
attained by keeping exactly the $k$ largest coefficients. Since the $k$ largest and the remaining $N-k$ coefficients partition the
total energy, \eqref{eq:best_k} and \eqref{eq:Ck} give
\begin{equation}
\varepsilon_k^2(d)=\|d\|_2^2\,\big(1-C_k(d)\big),
\label{eq:Ck_bestk}
\end{equation}
so $C_k$ is the normalized complement of a standard measure of
$k$-term compressibility. The three
descriptions are equivalent: $L_N$ as a curve, $C_k$ as the energy
retained, and $\varepsilon_k$ as the energy discarded at each $k$. It is
$\varepsilon_k$ that enters the risk bound of
Section~\ref{sec:theory}.

Consistent with this analysis, the numerical results reported in Section~\ref{sec:IV} show that composite product
transforms $\bm W_1 \bm W_2$ exhibit stronger coefficient disbalance and lower average MSE than single--basis wavelet transforms on the Donoho--Johnstone test signals.

\noindent \textbf{Scope of the sparsity improvement.}
The orthogonality of a composite transform guarantees energy preservation,
numerical stability, and perfect reconstruction, but it does not guarantee
greater coefficient concentration. Writing
\[
\uw d_1 = \bm W_2 \uw y,
\qquad
\uw d_2 = \bm W_1 \uw d_1 = \bm W_1 \bm W_2 \uw y ,
\]
the product improves top-$k$ compressibility relative to the inner transform
$\bm W_2$ only when $C_k(\uw d_2) > C_k(\uw d_1)$, and it improves upon both
component transforms only when
\[
C_k(\bm W_1 \bm W_2 \uw y)
\;>\;
\max\bigl\{\, C_k(\bm W_1 \uw y),\; C_k(\bm W_2 \uw y) \,\bigr\}.
\]
These conditions depend on the signal, on the order of composition, and on the
selected decomposition levels. Proposition~\ref{prop:nonuniversality} in
Appendix~\ref{sec:E} shows that strict improvement cannot hold universally, and
Fig.~\ref{fig:fig11} in Appendix~\ref{sec:D} characterizes empirically when it
does. The results reported below therefore demonstrate improvement for the
particular signals and transform configurations considered.

\begin{center}
\includegraphics[width=0.50\columnwidth]{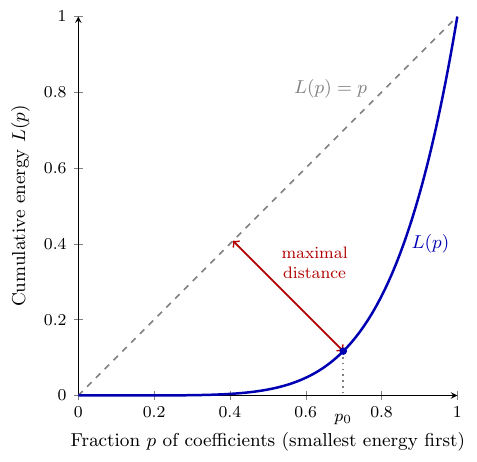}\\[4pt]
\end{center}
\noindent{\footnotesize\textbf{Explanatory Lorenz curve. The dashed diagonal $L(p)=p$ corresponds to evenly spread energy; a disbalanced representation lies farther from it, rising sharply only near $p=1$. The vertical drop at $p_0$ marks the location where the gap is largest, and the maximal (perpendicular) distance between the curve and the diagonal measures the concentration. }}

\subsubsection{\texorpdfstring{Product $\bm{W}_1\bm{W}_2$ with $\bm{W}_1$ being complex-valued}{Product W1W2 with W1 being complex-valued}}
\label{sec:3.1.1}

Let $y$ be the observed signal. Define two unitary wavelet transformation matrices: $\bm W_1$ generated
by a complex-valued filter (for example, complex Daubechies~6), and $\bm W_2$ generated by a real orthonormal
filter (such as Haar or real Daubechies). By a complex-valued filter we mean a wavelet filter whose taps are complex numbers; the resulting transform matrix is complex-valued and unitary rather than real orthogonal. The composite transform

\[
\bm W = \bm W_1 \cdot \bm W_2
\]
is unitary since it is a product of unitary matrices. We bring $y$ to the wavelet domain:

\[
\underset{\sim}{d} = \bm W y,
\]
where the complex-valued coefficients $d_k$ can be written as

\[
d_k = d_{\mathrm{Re},k} + i\, d_{\mathrm{Im},k},
\]
with magnitudes $|d_k|$.
To extract the true signal from the noisy signal, we apply a hard thresholding rule with universal threshold based on the magnitude before returning to the time domain. 
The denoising performance of product transforms incorporating complex-valued filters, with $\bm W_1$ being the complex DAUB~6 and
$\bm W_2$ being Haar, is evaluated
in Section~\ref{sec:IV} for the Doppler signal (see Table~\ref{tab:table1} and Fig.~\ref{fig:fig3}).

\subsection{Kronecker Product of Wavelet Matrices}
\label{sec:sec3.2}

The Kronecker product of two orthogonal (or unitary) wavelet matrices is itself orthogonal (unitary).
Indeed, if $\bm A$ and $\bm B$ are orthogonal, then
\begin{eqnarray*}
(\bm A \otimes \bm B)^\top (\bm A \otimes \bm B)
&=& (\bm A^\top \bm A) \otimes (\bm B^\top \bm B) \\
&=& \bm I \otimes \bm I
\;=\; \bm I,
\end{eqnarray*}
so the Kronecker product preserves energy and guarantees perfect reconstruction.
This stability makes Kronecker--structured transforms particularly attractive in wavelet shrinkage.
They naturally extend one--dimensional orthogonal wavelets to higher dimensions (for example, images or
multidimensional signals) while retaining sparsity--inducing properties.

As a concrete example, let $\bm A \in \mathbb{R}^{m \times n}$ denote a 2--D signal or image, and let
$\bm W_1 \in \mathbb{R}^{m \times m}$ and $\bm W_2 \in \mathbb{R}^{n \times n}$ be orthogonal wavelet
matrices. Then the 2--D scale--mixing wavelet transform is defined by
\begin{eqnarray}
\bm W_1 \bm A \bm W_2^\top,
\end{eqnarray}
and admits the vectorized identity
\begin{eqnarray}
\mathrm{vec}\!\left(\bm W_1 \bm A \bm W_2^\top\right)
&=& (\bm W_2 \otimes \bm W_1)\, \mathrm{vec}(\bm A),
\label{eq:kronecker_vec}
\end{eqnarray}
where $\mathrm{vec}(\cdot)$ denotes columnwise vectorization.
Thus, a 2--D transform can be represented as a Kronecker product transform applied to the vectorized
2--D object. Although orthogonal, this Kronecker (scale–mixing) construction is generally not a wavelet matrix in the strict filterbank sense; it is the separable orthogonal operator described above, built from two wavelet systems.

The block structure of Kronecker products also facilitates scalable shrinkage strategies: thresholding
can be applied along rows, columns, or subblocks, enabling both separable and blockwise shrinkage rules.
In quantum implementations, Kronecker structure aligns naturally with tensor--product Hilbert spaces,
allowing efficient circuit factorizations and making multidimensional shrinkage pipelines both
mathematically elegant and practically beneficial.

\subsection{Diagonal Block Wavelet Matrix}
\label{sec:sec3.3}

For wavelet matrices $\bm W_1$ and $\bm W_2$, the block matrix
\begin{eqnarray*}
\bm W = \mbox{diag}(\bm W_1,\bm W_2)
=
\left[\begin{array}{cc}
\bm W_1 & 0 \\
0 & \bm W_2
\end{array}\right]
\end{eqnarray*}
is orthogonal/unitary. This property gives an advantage for representing inhomogeneous signals.
As an illustrative example, consider the share price of a particular stock over a one year horizon. The stock price may
experience periods of smooth, stable behavior interspersed with sharp spikes and drops during
episodes of market volatility. In this setting, a single wavelet basis may not be sufficiently flexible to capture both regimes; it may oversmooth spikes and not remove enough noise from times of stability.
We therefore turn to a more customizable and adaptive method using the diagonal block wavelet matrix to investigate whether it improves performance relative to a single wavelet basis for signals with locally varying structures.

To reflect a heterogeneous signal, we consider a diagonal block wavelet matrix
\[
\bm W = \mathrm{diag}(\bm W_1,\bm W_2,...,\bm W_k),
\]
where each of the $k$ blocks corresponds to a wavelet basis suited to a particular structural component, such as oscillations, spikes, or discontinuous areas. This construction preserves orthogonality while allowing for different representations across signal segments. Like the Kronecker construction, the block–diagonal operator is generally not a wavelet matrix: applying distinct filters on segments that are disjoint breaks the shift–invariant (circulant) structure that a single two–channel filterbank requires across the block boundary, so no single filter generates the whole block-diagonal transform discussed here. Each diagonal block is itself an orthogonal multiscale transform applied to an entire segment, so the construction is segmentwise multiscale rather than pointwise. Its effectiveness depends on the segmentation: when the block boundaries can be placed at the signal's structural transitions, the construction limits the mixing of distinct structures across segments, whereas a boundary falling within a homogeneous region splits it across two bases.

A numerical evaluation of the diagonal block construction on a composite signal formed from the
four Donoho--Johnstone test functions is presented in Section~\ref{sec:IV} (Fig.~\ref{fig:fig4},
Table~\ref{tab:table2}).

\subsection{Similarity Transforms 
\texorpdfstring{$W^{\top} A W$}{WT A W}  when A is a wavelet Matrix}
\label{sec:sec3.4}

Given an orthogonal wavelet matrix $\bm W$ and a square matrix $A$, 
the similarity transform produces
\[
\bm B = \bm W^{\top} A \bm W.
\]
The similarity transform expresses the linear map $A$ in a new basis defined by $\bm W$, 
preserving properties such as eigenvalues and orthogonality (or unitarity, in which case the similarity transform produces
\(
\bm B = \bm W^{\dagger} A \bm W,
\)
where $\bm W^{\dagger}$ denotes the conjugate transpose).

When $A$ is a wavelet matrix, applying a similarity transform maintains orthogonality, 
but unlike a simple product $A \bm W$, it generally mixes scales and locations, so the resulting 
transform $\bm B$ may not retain the standard hierarchical structure of wavelet coefficients. 
Although similarity transforms preserve energy, they often lead to suboptimal performance 
for tasks such as thresholding or shrinkage, where sparsity and interpretability of 
wavelet coefficients are crucial. In our experiments, similarity–transformed wavelets 
were somewhat inferior to pairwise products of wavelet matrices and typically produced 
higher mean–squared error (MSE) after shrinkage compared to the original $ \bm A$, 
indicating that orthogonality alone does not guarantee improved alignment for sparsity 
\citep{goel1995}.

\subsection{Denoising Guarantees for Composite Transforms}
\label{sec:theory}
 
The four constructions above share a single foundation: each is
orthogonal by construction. Consequently, the denoising
guarantees of classical wavelet thresholding extend to all of them
uniformly in a real-valued setting. We make this explicit here. Recall from Section~\ref{sec:sec3.1}
that given $\uw y = \uw s + \uw\varepsilon$, the composite coefficients
split as $\uw d = \bm W \uw y = \bm\theta + \bm W\uw\varepsilon$ with
$\bm\theta = \bm W\uw s$, and  Lemma~\ref{lem:noise} characterizes the
noise term $\bm W\uw\varepsilon$.

\begin{corollary}[Validity of the universal threshold]
\label{cor:universal}
Let $\bm W$ be orthogonal and let $\lambda = \sigma\sqrt{2\log N}$. Then
\[
\Pr\!\Big(\max_{1\le i\le N} |(\bm W\uw\varepsilon)_i| > \lambda \Big)
\;\le\; \frac{1}{\sqrt{\pi \log N}} \;\xrightarrow[N\to\infty]{}\; 0 .
\]
Hence the universal threshold removes pure noise with probability tending
to one in any orthonormal basis. 
\begin{proof}
Follows from Lemma~\ref{lem:noise} and a Gaussian tail bound; see
Appendix~\ref{sec:C}.
\end{proof}
\end{corollary}

Corollary~\ref{cor:universal}  shows that the threshold removes noise. We now bound the error of the denoised signal itself. The
following bound then follows from the classical theory of thresholding.

\begin{corollary}[Oracle inequality]
\label{cor:oracle}
Let $\bm W$ be orthogonal and $\bm\theta = \bm W\uw s$. Let
$\uw d^{\,*}$ denote the hard thresholding of the coefficients
$\uw d = \bm W\uw y$ at $\lambda = \sigma\sqrt{2\log N}$, and let
$\hat{\uw s} = \bm W^\top \uw d^{\,*}$ be the reconstructed estimate. Then
\[
\mathbb{E}\,\|\hat{\uw s} - \uw s\|_2^2
\;\le\; \Lambda_N\Big(\sigma^2
       + \sum_{i=1}^{N} \min(\theta_i^2,\sigma^2)\Big),
 \Lambda_N \sim 2\log N .
\]
\end{corollary}
\begin{proof}
By Lemma~\ref{lem:noise}, the noise coefficients
$z_i = (\bm W\uw\varepsilon)_i$ are i.i.d.\ $\mathcal{N}(0,\sigma^2)$, so
the observed coefficients satisfy $d_i = \theta_i + z_i$ with
$\bm\theta = \bm W\uw s$. Now due to orthogonality of $\bm W$, we have preservation of lengths $\|\hat{\uw s}-\uw s\|_2^2 = \|\uw d^{\,*}-\bm\theta\|_2^2$ (Parseval's identity). The estimation
of $\bm\theta$ from $d_i = \theta_i + z_i$ by hard thresholding is exactly the setting of \citet[Theorem~4]{DonohoJohnstone1994}, which gives the
stated oracle risk bound.
\end{proof}

Note that the bound holds for any orthogonal $\bm W$  with the same constant $\Lambda_N$. The only term that changes
from one transform to another is the oracle risk  from Corollary~\ref{cor:oracle} :

 \begin{equation}
R(\bm\theta;\sigma):=\sum_{i=1}^{N}\min\!\big(\theta_i^2,\sigma^2\big),
\qquad \bm\theta=\bm W\uw s .
\label{eq:oracle_term}
\end{equation}

Let $E:=\|\uw s\|_2^2$. Since $\bm W$ is orthogonal, Parseval's
identity gives $\|\bm\theta\|_2^2=\|\uw s\|_2^2=E$ regardless of which
$\bm W$ is used; this is what allows $C_k$ and $\varepsilon_k$ of two
different transforms to be compared. The following identity expresses \eqref{eq:oracle_term} in terms of $\varepsilon_k$ and $C_k$.

\begin{proposition}[Oracle risk as a best-$k$-term trade-off]
\label{prop:identity}
For any $\bm\theta$ with $\|\bm\theta\|_2^2=E$ and any $\sigma>0$,
\begin{equation}
\begin{split}
R(\bm\theta;\sigma)
=\min_{0\le k\le N}\big[k\sigma^2+\varepsilon_k^2(\bm\theta)\big] \\
=\min_{0\le k\le N}\big[k\sigma^2+\big(1-C_k(\bm\theta)\big)E\big],
\end{split}
\label{eq:identity}
\end{equation}
with $C_0:=0$. The minimum is attained at
$k^\ast=\#\{i:\theta_i^2>\sigma^2\}$.
\end{proposition}

\begin{proof}
Fix arbitrary $k$ between $0$ and $N$. Bounding the $k$ terms of \eqref{eq:oracle_term} with the largest
$|\theta_i|$ by
$\sigma^2$ and the remaining $N-k$ terms by $\theta_i^2$ gives
\[
R(\bm\theta;\sigma)\;\le\;k\sigma^2+\sum_{j>k}|\theta|_{[j]}^2
\;=\;k\sigma^2+\varepsilon_k^2(\bm\theta).
\]
Minimizing over $k$ yields ``$\le$'' in \eqref{eq:identity}. At $k=k^\ast$
both bounds hold with equality coordinate-wise, since
$\min(\theta_i^2,\sigma^2)=\sigma^2$ for exactly the $k^\ast$ largest
coefficients and $\min(\theta_i^2,\sigma^2)=\theta_i^2$ for the rest; hence equality. The second
form of \eqref{eq:identity} is \eqref{eq:Ck_bestk}.
\end{proof}

Identity \eqref{eq:identity} is the bias--variance trade-off of nonlinear
approximation \citep[Ch.~9]{Mallat2009wavelet}: the oracle pays $\sigma^2$
of variance for each coefficient it retains and $\varepsilon_k^2$ of bias
for those it discards, and the transform enters only through
$\varepsilon_k$ (which is equivalently $C_k$ and equivalently the Lorenz curve). It
also shows why concentration at small $k$ is what matters for denoising:
the minimizer $k^\ast$ is the number of coefficients above the noise
level, which is small for a sparse representation, so only the low-$k$
portion of the Lorenz curve is active in \eqref{eq:identity}.

Identity \eqref{eq:identity} turns the comparison of two transforms into
a comparison of their $\varepsilon_k$ at every $k$, and hence of their
Lorenz curves. The following theorem makes this exact.

\begin{theorem}[Lorenz dominance and oracle-risk ordering]
\label{thm:lorenz_oracle}
Let $\bm W_a$ and $\bm W_b$ be orthogonal, and let
$\bm\theta^a=\bm W_a\uw s$, $\bm\theta^b=\bm W_b\uw s$, with Lorenz
curves $L_N^a$, $L_N^b$. The following are equivalent:
\begin{enumerate}
\item[(i)] $L_N^a(p)\le L_N^b(p)$ for all $p\in[0,1]$;
\item[(ii)] $C_k(\bm\theta^a)\ge C_k(\bm\theta^b)$ for every
      $k=1,\dots,N$;
\item[(iii)] $\varepsilon_k^2(\bm\theta^a)\le\varepsilon_k^2(\bm\theta^b)$
      for every $k=0,\dots,N$;
\item[(iv)] $R(\bm\theta^a;\sigma)\le R(\bm\theta^b;\sigma)$ for every
      $\sigma>0$.
\end{enumerate}
Under any of these conditions the bound of Corollary~\ref{cor:oracle}
for $\bm W_a$ is no larger than that for $\bm W_b$ at every noise level.
\end{theorem}

\begin{proof}
(i)$\Leftrightarrow$(ii):  By \eqref{eq:Ck}, at the grid points
$p=m/N$,\( L_N(m/N)=1-C_{N-m}.\)
Since $L_N$ is constant on each interval
$[m/N,(m+1)/N)$, it is determined by its values at these grid points.
Thus, pointwise ordering of the Lorenz curves is equivalent to ordering
of every $C_k$.
(ii)$\Leftrightarrow$(iii): by \eqref{eq:Ck_bestk}, since both
transforms have total energy $E$.
(iii)$\Rightarrow$(iv): by Proposition~\ref{prop:identity}, each
candidate $k\sigma^2+\varepsilon_k^2(\bm\theta^a)$ in the minimum is at
most the corresponding candidate for $\bm\theta^b$, so the minimum is
smaller.
(iv)$\Rightarrow$(ii): write $c=\sigma^2$ and $t =\theta_i^2$. Since $\min(t,c)=t-(t-c)_+$
and both energy vectors sum to $E$, (iv) is equivalent to
\begin{equation}
\sum_{i=1}^N\big((\theta_i^a)^2-c\big)_+
\;\ge\;\sum_{i=1}^N\big((\theta_i^b)^2-c\big)_+
\qquad\text{for all }c>0 .
\label{eq:hinge}
\end{equation}
Fix $k$ and take $c=|\theta^a|_{[k]}^2$. If $c>0$, the left side of
\eqref{eq:hinge} equals $\sum_{j\le k}|\theta^a|_{[j]}^2-kc$ exactly,
since the remaining terms are non-positive and clipped to zero, while
the right side is at least $\sum_{j\le k}|\theta^b|_{[j]}^2-kc$, since
clipping can only increase a sum. Hence
$\sum_{j\le k}|\theta^a|_{[j]}^2\ge\sum_{j\le k}|\theta^b|_{[j]}^2$,
which is (ii). If $c=0$ then $C_k(\bm\theta^a)=1$ and (ii) holds
trivially.
\end{proof}

\begin{remark}
Condition (ii) states that the energy vector $((\theta_i^a)^2)_i$
majorizes $((\theta_i^b)^2)_i$ in the sense of
Hardy--Littlewood--P\'olya \citep{HardyLittlewoodPolya1934}, and (iv) is
the statement that $\bm x\mapsto\sum_i\min(x_i,\sigma^2)$ is
Schur-concave; the theorem is thus an instance of majorization theory
\citep[Ch.~3--4]{MarshallOlkinArnold2011}. The proof above is included to
keep the paper self-contained.
\end{remark}

\begin{remark}[Clean and noisy coefficients]
Theorem~\ref{thm:lorenz_oracle} concerns the clean coefficients
$\bm\theta=\bm W\uw s$. For noisy coefficients
$\uw d=\bm\theta+\bm W\uw\varepsilon$, Lemma~\ref{lem:noise} gives
$\mathbb E\,d_i^2=\theta_i^2+\sigma^2$ in every orthogonal basis, so at
the level of expected energies the noise adds the same floor to every
coefficient and preserves the ordering (ii). This does not imply that
(i) holds for each noise realization. Accordingly, the clean block of
Table~\ref{tab:sparsity_metrics} is the quantity to which the theorem
applies, and the noisy Lorenz curves of Fig.~\ref{fig:fig1} are
empirical diagnostics of the coefficients the threshold acts on.
\end{remark}

The Lorenz curve is a visualization of the oracle term, and
Theorem~\ref{thm:lorenz_oracle}  is the precise statement: a curve far
below the diagonal means most coefficients hold little energy, so most of
the terms $\min(\theta_i^2,\sigma^2)$ are small and the oracle risk is low. The
Lorenz curves in Fig.~\ref{fig:fig1} therefore show
the term that controls the risk bound, which is why transforms with
stronger disbalance can be expected to give lower MSE under the same threshold. The bound does not by itself establish that composite transforms outperform single--basis transforms; it shows how they can, through a reduction in oracle risk. 

\noindent
 Additionally, Proposition~\ref{prop:nonuniversality}
 shows that this reduction cannot hold for every signal.  Corollary~\ref{cor:universal}  depends only on the orthogonality of $\bm W$
and not on $\uw s$, so the universal threshold remains valid whether or not the
signal is sparse. What sparsity affects is instead the oracle term of
Corollary~\ref{cor:oracle}, which can be large when $\bm\theta$ is dense. Whether $\bm\theta=\bm W\uw s$ is dense depends on both the signal and the
transform. Thus, composite and single-basis transforms differ in risk only where one of
them represents the signal more sparsely than the other. The simulations of Section~\ref{sec:IV}  show that this improvement is
realized in the cases studied; the theorem orders the risk bounds and does
not by itself guarantee the ordering of the realized MSE.

%%%%%%%%%%%%%%%%%%%%%%%%%%%%%%%%%%%%%%%%%%%%%%%%%%%%
% SIMULATIONS AND EXAMPLES
%%%%%%%%%%%%%%%%%%%%%%%%%%%%%%%%%%%%%%%%%%%%%%%%%%%%
\section{Simulation Analysis of Performance of Wavelet-like Matrices in Denoising}
\label{sec:IV}
% Bridge from theory to experiments
\noindent
Having established the structural properties of composite wavelet-like matrices, we now assess their practical impact on denoising performance through controlled simulation studies.

\label{sec:sec4}

In this section, we quantify the denoising performance of several composite wavelet--like
orthogonal operators introduced in Section~\ref{sec:sec3}.
Throughout, the processing
pipeline is held fixed: we transform the noisy signal by an orthogonal/unitary matrix,
apply coefficient shrinkage using the universal threshold, and reconstruct by using the
inverse transform ($\bm W^\top$ in the real orthogonal case and
$\bm W^\dagger$ in the complex unitary case). This isolates the effect of the transform itself.
We performed $M = 200$ independent Monte Carlo replications and used mean--squared error (MSE) as the
primary performance metric. The formula for average MSE (AMSE) is as follows:  \begin{eqnarray}
\mathrm{AMSE} = \frac{1}{M N}\sum_{j=1}^M \sum_{i=1}^N \big(s_{\text{test},i} - \hat{s}_{i,j}\big)^2,
\label{eq:amse}
\end{eqnarray}
where
$\hat{s}_{i,j}$ is an estimate of the $i$th component of the signal at the $j$th Monte Carlo iteration, and $s_{\text{test},i}$ is the $i$th component of the signal. The composite transforms examined below are the product
transform $\bm W_1 \bm W_2$, the Kronecker product $\bm W_k \otimes \bm W_l$, and the
similarity transform $\bm W_1^{\top} \bm W_2 \bm W_1$, each compared against single--basis
wavelet transforms.

\medskip
\noindent
\textbf{Simulation framework.}
All simulation experiments are conducted using the four canonical Donoho--Johnstone test signals
(Doppler, Blocks, HeaviSine, and Bumps), representing oscillatory, discontinuous,  and locally spiked
features. We fix the signal length at $N=1024$ throughout. Unless otherwise stated, a decomposition level of $L=3$ is used for length $N = 1024$ signals because it gives strong coefficient concentration and near-minimal denoising MSE for Doppler, Blocks, and HeaviSine, with no further benefit from deeper decomposition; $L=1$ is used for the localized Bumps signal, whose performance is best at a finer level.

\noindent
In each replication, Gaussian white noise $\varepsilon \sim \mathcal{N}(0,I_N)$ is added
to a rescaled clean signal  $x$ so that the noise standard deviation is $\sigma=1$ and is known by construction. The rescaling of the signal is chosen to achieve a prescribed signal-to-noise ratio, which, following standard engineering convention, we define
as a ratio of powers,
\begin{eqnarray}
P_{\mathrm{signal}} = \frac{1}{N}\sum_{i=1}^{N} x_i^2 = \frac{\|x\|_2^2}{N},
\qquad
P_{\mathrm{noise}} = \sigma^2,
\label{eq:snr_power}
\end{eqnarray}  and report it in decibels as
\begin{eqnarray}
\mathrm{SNR}_{\mathrm{in}}
= 10\log_{10}\!\left(\frac{P_{\mathrm{signal}}}{P_{\mathrm{noise}}}\right).
\label{eq:snr_db}
\end{eqnarray} Since $\mathbb{E}\|\varepsilon\|_2^2 = N\sigma^2$, the power ratio in
\eqref{eq:snr_db} coincides with the corresponding energy ratio. Explicitly, for a
target $\mathrm{SNR}_{\mathrm{in}}$ and a raw signal $x_0$ we set
\[
x = \left(\frac{\sigma^2\,10^{\,\mathrm{SNR}_{\mathrm{in}}/10}}
               {\tfrac{1}{N}\|x_0\|_2^2}\right)^{1/2} x_0 ,
\]
so that \eqref{eq:snr_db} holds by construction.  We set $\mathrm{SNR}_{\mathrm{in}}=5$~dB in the main simulations (unless otherwise stated), a
moderate noise level at which the effect of the transform on denoising
performance is more visible. 

 Because $\sigma$ is known, we use $\lambda = \sigma\sqrt{2\log N}$ directly, which isolates the effect of the transform from that of variance estimation. When $\sigma$ is unknown in practice, a standard robust wavelet estimate is $\hat\sigma = \mathrm{median_k}(|d_{\mathrm{fine,k}}|)/0.6745$, using approximately zero-centered finest-scale detail coefficients,
and the resulting threshold can be applied unchanged in the composite domain.

Throughout our simulations and applications, we employ universal thresholding as a single, fixed shrinkage rule to isolate the effect of the transform itself and enable transparent comparisons across bases and composite constructions. Because the proposed transforms remain orthogonal/unitary, the same analysis pipeline applies without modification to alternative shrinkage schemes.

\medskip
\noindent
\textbf{Overview of simulation cases.}
We consider three classes of experiments: (i) real--valued product transforms evaluated across the
four Donoho--Johnstone signals; (ii) product and related composite transforms incorporating complex-valued
filters for the Doppler signal; and (iii) adaptive diagonal block constructions combining multiple
wavelet bases for the four signals.

\begin{figure}[!t]
  \centering
  \includegraphics[width=0.7\columnwidth]{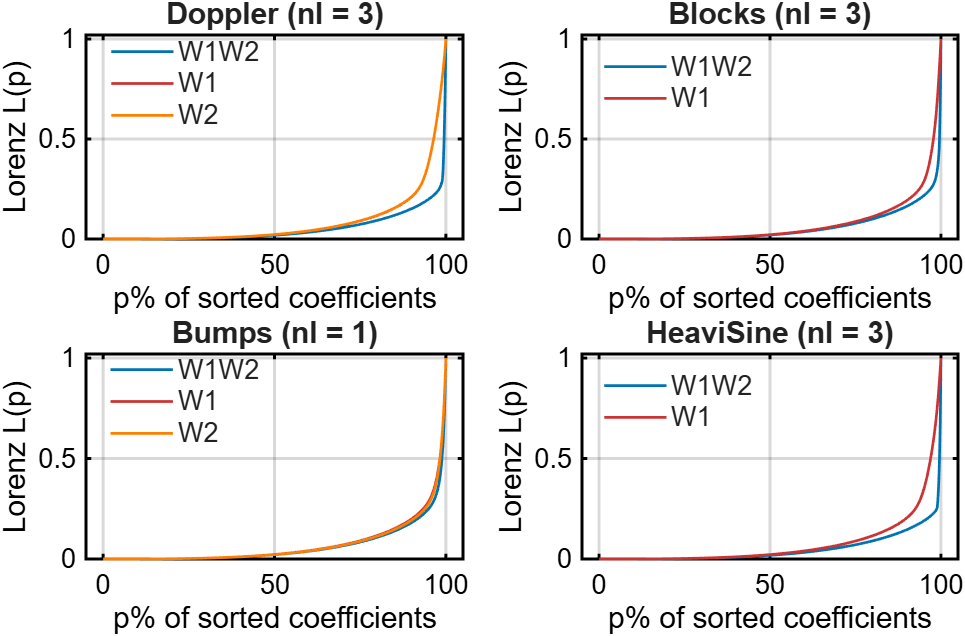}
  \caption{Lorenz Curves for four Donoho--Johnstone test signals  at $\mathrm{SNR}_{\mathrm{in}}=5$~dB. Product matrix $\bm W_1 \bm W_2$
  yields a substantially more disbalanced energy distribution than single--basis wavelet transforms.
  The number of levels of decomposition is $3$ for three of the four signals, with one level of
  decomposition showing clearest disbalance for the Bumps signal. For HeaviSine and Blocks signals, $\bm W_1=\bm W_2$ so only $\bm W_1$ is shown.}
  \label{fig:fig1}
\end{figure}

Fig.~\ref{fig:fig1} summarizes coefficient--energy concentration via Lorenz curves. Recall from the discussion of Lorenz curve (\ref{eq:lorenz}),
that a stronger departure from the $45$--degree line indicates greater
disbalance (energy concentration), a criterion linked to improved denoising efficacy under
thresholding \citep{goel1995}. Across the Donoho--Johnstone benchmark signals, the product transform
$\bm W_1 \bm W_2$ consistently produced more concentrated energy distributions than the components as single--basis
transforms, supporting the premise that composite orthogonal operators can yield sparser coefficient representations.

\begin{figure}[!t]
  \centering
  \includegraphics[width=0.7\columnwidth]{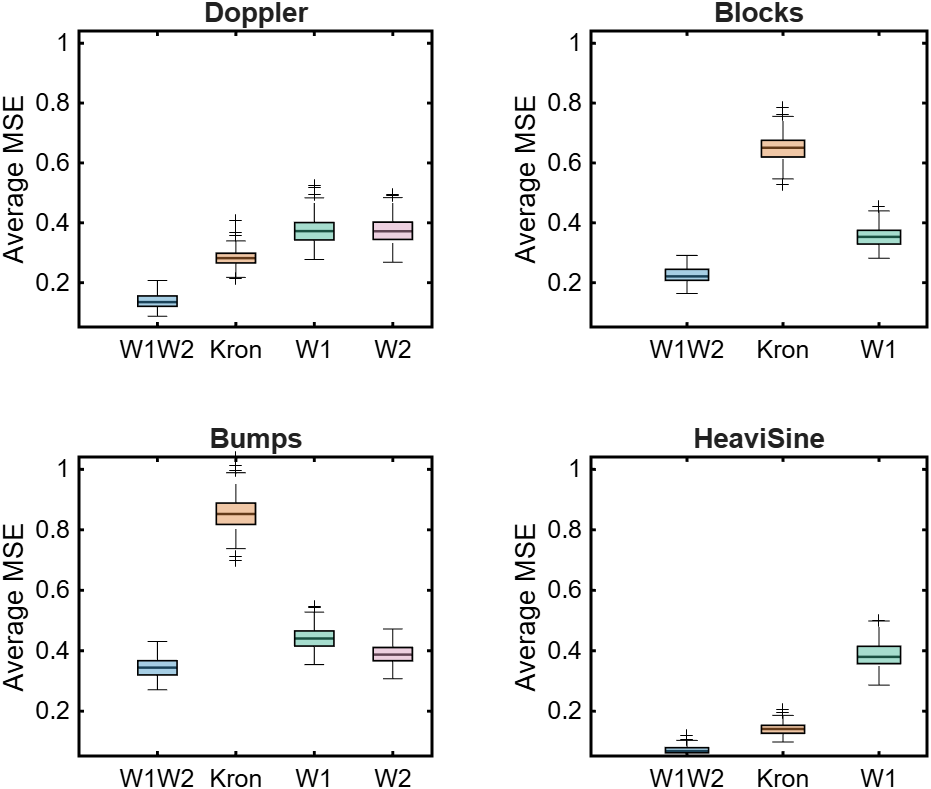}
  \caption{Product transform matrices $\bm W_1 \bm W_2$ (chosen respectively for each signal) exhibit
  lower AMSE over $200$ simulations for Donoho--Johnstone signals compared to single--basis transforms,  at
$\mathrm{SNR}_{\mathrm{in}}=5$~dB. Kronecker product performs similarly or better than single--basis transforms on two of the four signals.}
  \label{fig:fig2}
\end{figure}

Fig.~\ref{fig:fig2} reports AMSE summaries over the same $200$ replications. The product transforms
$\bm W_1 \bm W_2$ achieve the lowest AMSE in all four Donoho--Johnstone benchmark signals, with the
Kronecker construction competitive in selected cases. These results align with the Lorenz curve
evidence: transforms that induce stronger energy concentration in the coefficient domain tend to
yield better denoising under identical thresholding. To assess robustness to noise intensity, we repeated these one-dimensional
experiments over $\mathrm{SNR}_{\mathrm{in}}\in\{-5,0,5,8.45\}$~dB, spanning noise-dominated to signal-dominated regimes under the power definition ~\eqref{eq:snr_db}. 
The product transform retains its advantage over the single--basis transforms
across the full range, attaining the lower AMSE at every level; the
MSE-versus-SNR curves and the underlying values are reported in
Appendix~\ref{sec:D}.

Table~\ref{tab:sparsity_metrics} expresses the sparsity comparison of
Fig.~\ref{fig:fig1} numerically, using the metrics of
Section~\ref{sec:sec3.1} and reporting each under two conditions. The noisy
columns are computed on $\bm W(\uw s + \uw\varepsilon)$ at $\mathrm{SNR}_{\mathrm{in}}=5$~dB,
averaged over the same $200$ realizations used for Fig.~\ref{fig:fig1}; these are
the coefficients the threshold acts on. The clean columns are computed on
$\bm\theta = \bm W \uw s$, which is what enters the oracle risk of
Corollary~\ref{cor:oracle}. In every case the product transform gives the larger
value, confirming that a Lorenz curve lying farther below the diagonal
corresponds to a strictly sparser representation as measured by $C_k$.

The gap is widest at small $k$, and it is at small $k$ that the comparison is
informative.  By Lemma~\ref{lem:noise} the transformed noise is i.i.d.\
$\mathcal{N}(0,\sigma^{2})$ in every orthonormal basis: a transform may
concentrate the signal, but it cannot concentrate the noise. As $k$ increases, a
growing share of the energy counted by $C_k$ comes from coefficients that are
predominantly noise, and since that contribution is identical across bases, the
metrics of different transforms necessarily converge. Thresholding at
$\sigma\sqrt{2\log N}$ retains only coefficients above the noise level, so the
reconstruction error is governed by the largest coefficients.

\begin{table}[!t]
\centering
\caption{Sparsity metrics for the four Donoho--Johnstone benchmark signals
($N=1024$). $C_{k\%}$ is the fraction of coefficient energy carried by the
largest $k\%$ of coefficients, as defined in Section~\ref{sec:sec3.1}. Left
block: noisy coefficients at $\mathrm{SNR}_{\mathrm{in}}=5$~dB under the power
definition~\eqref{eq:snr_db}, averaged over $200$ realizations and matching the
conditions of Fig.~\ref{fig:fig1}.  Right block: clean-signal coefficients $\bm W \uw s$, the quantity entering the oracle risk of Corollary~\ref{cor:oracle}; these are scale
invariant and so do not depend on the noise level. The largest value in each column is in bold.}
\label{tab:sparsity_metrics}
\resizebox{\columnwidth}{!}{%
\begin{tabular}{ll S[table-format=1.4] S[table-format=1.4] S[table-format=1.4]
                   S[table-format=1.4] S[table-format=1.4] S[table-format=1.4]}
\toprule
& & \multicolumn{3}{c}{\textbf{Noisy, $5$~dB}} & \multicolumn{3}{c}{\textbf{Clean}} \\
\cmidrule(lr){3-5}\cmidrule(lr){6-8}
\textbf{Signal} & \textbf{Transform}
  & {$C_{1\%}$} & {$C_{5\%}$} & {$C_{10\%}$}
  & {$C_{1\%}$} & {$C_{5\%}$} & {$C_{10\%}$} \\
\midrule
\multirow{3}{*}{Doppler}
 & $\bm W_1$ (Daub5)          & 0.1829 & 0.6150 & 0.7888 & 0.2108 & 0.7690 & 0.9886 \\
 & $\bm W_2$ (Symm4)          & 0.1830 & 0.6144 & 0.7878 & 0.2109 & 0.7671 & 0.9853 \\
 & $\bm W_1 \bm W_2$          & \bfseries 0.6999 & \bfseries 0.8005 & \bfseries 0.8481
                              & \bfseries 0.9177 & \bfseries 0.9987 & \bfseries 0.9999 \\
\midrule
\multirow{2}{*}{Blocks}
 & $\bm W_1 = \bm W_2$ (Haar) & 0.2845 & 0.7129 & 0.8074 & 0.3475 & 0.9081 & 0.9942 \\
 & $\bm W_1 \bm W_2$          & \bfseries 0.6444 & \bfseries 0.7827 & \bfseries 0.8361
                              & \bfseries 0.8412 & \bfseries 0.9835 & \bfseries 0.9981 \\
\midrule
\multirow{3}{*}{Bumps}
 & $\bm W_1$ (Haar)           & 0.3607 & 0.7133 & 0.8012 & 0.4672 & 0.9074 & 0.9800 \\
 & $\bm W_2$ (Daub3)          & 0.3567 & 0.7249 & 0.8075 & 0.4571 & 0.9238 & 0.9871 \\
 & $\bm W_1 \bm W_2$          & \bfseries 0.4657 & \bfseries 0.7497 & \bfseries 0.8191
                              & \bfseries 0.6088 & \bfseries 0.9537 & \bfseries 0.9935 \\
\midrule
\multirow{2}{*}{HeaviSine}
 & $\bm W_1 = \bm W_2$ (Symm4)& 0.2282 & 0.6530 & 0.7951 & 0.2803 & 0.8252 & 0.9921 \\
 & $\bm W_1 \bm W_2$          & \bfseries 0.7250 & \bfseries 0.8104 & \bfseries 0.8543
                              & \bfseries 0.9489 & \bfseries 0.9998 & \bfseries 1.0000 \\
\bottomrule
\end{tabular}%
}
\end{table}
\medskip
\noindent
{\bf Complex-valued filters within product transforms.}
Building on Section~\ref{sec:3.1.1}, we next incorporate complex-valued filters in product transforms and
evaluate performance in a unitary setting using the Doppler signal.

\noindent
For this signal, we choose $\bm W_1$ to be the $L$--level unitary transform matrix generated by the
complex-valued Daubechies~6 filter, and $\bm W_2$ to be the $L$--level unitary transform matrix generated by
the (real) Haar filter. We compare three composite wavelet--like matrices: the product matrix
$\bm W_1 \bm W_2$, the similarity transform $\bm W_1^{\dagger}\bm W_2 \bm W_1$, and the Kronecker product
$\bm W_k \otimes \bm W_l$. Here, $\bm W_k$ and $\bm W_l$ are
unitary transforms generated by complex-valued Daubechies~6 and real Haar, respectively. The results and details of the experiment are shown in Fig.~\ref{fig:fig3} and Table~\ref{tab:table1}. Over $200$
simulations  and $\mathrm{SNR}_{\mathrm{in}}=5$~dB, the Kronecker and product transforms induce the lowest AMSE (and smallest variance),
with the similarity transform  performing better than one of the two single-basis transforms.

\begin{table}[!t]
\centering
\caption{Comparison of average MSE and variance of MSE over $200$ simulations of the Doppler signal
with $\mathrm{SNR}_{\mathrm{in}}=5$~dB. Here $\bm W_1$ is generated using complex Daubechies~6 filter, and $\bm W_2$ by
Haar. Kronecker product is between a complex-valued DAUB~6 wavelet transform $\bm W_k$ of size $2^7$ and a real
Haar transform $\bm W_l$ of size $2^3$, producing a $1024 \times 1024$ unitary transformation matrix.}
\begin{tabular}{lcc}
\hline
\textbf{Method} & \textbf{Average MSE} & \textbf{Variance of MSE} \\
\hline
Kronecker Product   & \textbf{0.1800} &  \textbf{0.0004}   \\
$\bm W_1 ^\dagger \bm W_2 \bm W_1$ Similarity &  0.3489 & 0.0017  \\
$\bm W_1 \bm W_2$ Product        & \textbf{0.1846}  & \textbf{0.0002} \\
$\bm W_1$ (complex D6)     & 0.3377 & 0.0016 \\
$\bm W_2$ (Haar)           & 0.4354  & 0.0017 \\
\hline
\end{tabular}
\label{tab:table1}
\end{table}

\begin{figure}[!t]
  \centering
  \includegraphics[width=0.7\columnwidth]{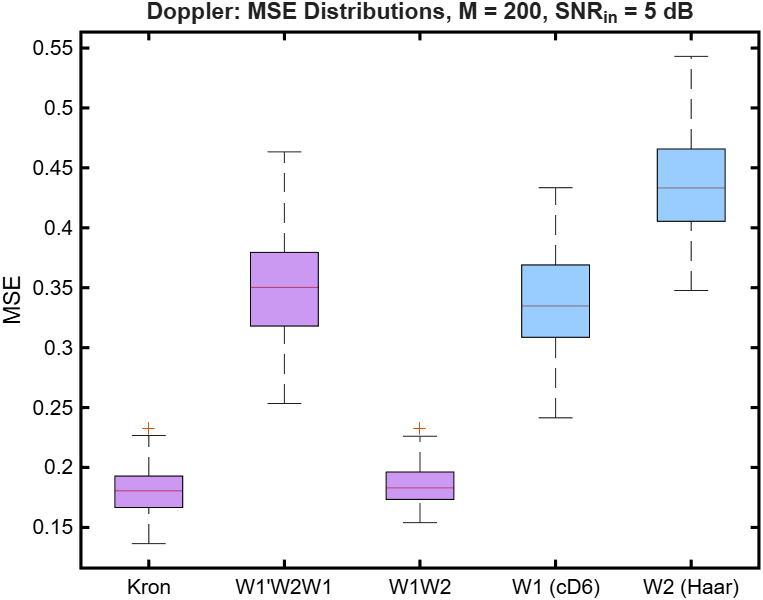}
  \caption{Comparison of Kronecker, similarity and product transforms against single--basis transforms for the Doppler signal ($N=1024$)  at $\mathrm{SNR}_{\mathrm{in}}=5$~dB; $\bm W_1 = $ complex-valued DAUB 6, $\bm W_2 = $ real Haar. After calculating MSE over $200$ simulations, the product transform  $\bm W_1 \bm W_2$ and Kronecker $\bm W_k \otimes \bm W_l$ produce the lowest average MSE and variance.
  These composite matrices have superior performance to the single--basis methods. Table~\ref{tab:table1} contains the exact values.}
  \label{fig:fig3}
\end{figure}

\medskip
\noindent
{\bf Diagonal block (adaptive) transforms.}
We also consider composite transforms constructed via diagonal block matrices as in
Section~\ref{sec:sec3.3}. We begin by forming a test signal by combining Doppler, Blocks, HeaviSine, and Bumps functions with
weights $\{1,0.2,0.1,0.2\}$ to introduce distinct structural features across intervals:

\[
s_{\text{test}} =
\big[s_{\text{Doppler}},\; 0.2\,s_{\text{Blocks}},\; 0.1\,s_{\text{HeaviSine}},\; 0.2\,s_{\text{Bumps}}\big].
\]
For each of the four Donoho--Johnstone test signals, there is a corresponding wavelet most used in the published literature: Symmlet~4 for Doppler, Haar for Blocks, Daubechies~4 (8 tap coefficients) for HeaviSine, and Daubechies~3 (6 tap coefficients) for Bumps~\citep{DonohoJohnstone1994}.
Our adaptive wavelet matrix is constructed as
\begin{eqnarray*}
\bm W = \mbox{diag}(\bm W_1,\bm W_2,\bm W_3,\bm W_4),
\end{eqnarray*}
with one block per segment of the signal, respectively.
This structure is no longer necessarily a classical wavelet matrix, yet orthogonality is preserved.

As with product and Kronecker matrices, for block diagonal  the same pipeline as above is applied ($\mathrm{SNR}_{\mathrm{in}}=5$~dB, hard thresholding at $\sigma\sqrt{2\log N}$, inverse by $\bm W^\top$); the MSE summaries are in Fig.~\ref{fig:fig4}  and Table~\ref{tab:table2}.

\begin{figure}[!t]
  \centering
  \includegraphics[width=0.7\columnwidth]{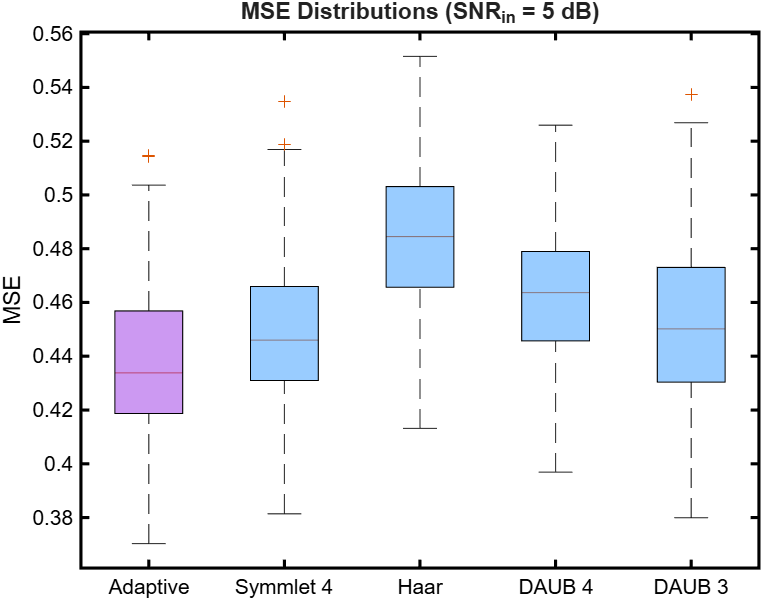}
  \caption{Side--by--side boxplots displaying distributions of MSEs ($200$ simulations)   for a combined signal of four Donoho--Johnstone signals,  at $\mathrm{SNR}_{\mathrm{in}}=5$~dB. The red bar corresponds to the median MSE value, which due to symmetry is close to the mean. Applying the adaptive wavelet matrix results in a lower average MSE than single--basis transforms (Table~\ref{tab:table2}).}
  \label{fig:fig4}
\end{figure}

\begin{table}[!t]
\centering
\caption{Comparison of adaptive diagonal block wavelet matrix against single--basis matrices on a
signal composed of a linear combination of four Donoho--Johnstone signals. Displayed are the average
and variance of MSE over $200$ simulations ( $\mathrm{SNR}_{\mathrm{in}}=5$~dB).}
\begin{tabular}{lcc}
\hline
\textbf{Method}    & \textbf{Average MSE} & \textbf{Variance} \\
\hline
Adaptive           & \textbf{0.4380}      & 0.0008 \\
Symmlet 4          & 0.4493               & 0.0007 \\
Haar               & 0.4845               & 0.0007 \\
Daubechies 4       & 0.4615               & 0.0007 \\
Daubechies 3       & 0.4515               & 0.0008 \\
\hline
\end{tabular}
\label{tab:table2}
\end{table}

\medskip
\noindent

\textbf{Choice of wavelet filters and composition order.}
Single--basis transforms $\bm W_1$ and $\bm W_2$ are selected following standard practice in the
wavelet denoising literature, where specific bases are known to provide effective representations
for particular signal features~\citep{DonohoJohnstone1994}. 
Composite product transforms $\bm W_1 \bm W_2$ are formed using these
fixed filter pairs\footnote{This refers to the primary simulation cases; Appendix~\ref{sec:D} separately reports an exhaustive sweep over all 64 ordered filter pairs.}. Empirical checks confirmed that reversing
the order of composition ($\bm W_2 \bm W_1$) consistently resulted in slightly inferior performance, and therefore such cases are omitted.

\medskip
\noindent

\textbf{On the similarity transform.}
In several settings, the similarity transform $\bm W_1^{\top} \bm W_2 \bm W_1$ yields noticeably
larger MSE values than the remaining transforms, which compresses the scale of boxplots and obscures
comparative differences. For clarity of presentation, the similarity transform is omitted from
selected Lorenz curve and boxplot displays, and is reported explicitly only in cases where its
performance is somewhat competitive.

\medskip
\noindent
\textbf{Remark}: Throughout Sections~\ref{sec:sec3} and~\ref{sec:sec4}, universal thresholding is taken as a baseline to highlight the effectiveness of composite wavelet-like transformations compared to single-basis transforms. Alternative shrinkage methods, such as Bayesian shrinkage via global-parameter BAMS \citep{VidakovicRuggeri2001}
and ABE thresholding \citep{figueiredo2001wavelet},
produce results consistent with those obtained under universal thresholding, and can be found in the supplementary materials
\footnote{\url{https://github.com/rrkulkarni108/CompositeWavelets}}. We note further that the universal threshold is in fact the rule that adaptive procedures themselves select in this regime: the hybrid SureShrink scheme of~\citet{DonohoJohnstone1995} reverts to $\sigma\sqrt{2\log N}$ at resolution levels judged sparse, applying SURE only where coefficients are sufficiently dense.

\section{Applications}
\label{sec:V}
% Bridge from simulations to data
\noindent
We now apply the composite transform framework developed in Section~\ref{sec:sec3} to real data. Having established its performance advantages in the simulation study of Section~\ref{sec:sec4}, we examine its behavior in practical settings, where interpretability and structural preservation are critical.

\subsection{Image Denoising: \textit{Barbara} Test Image}
\label{sec:sec5A}

Wavelet image denoising is a classical setting in which single--basis wavelets often perform extremely well. 
When an appropriate basis and decomposition depth are selected, and thresholding is applied properly,
reconstruction can be remarkably close to the original image. This strength, however, has a natural
limitation: single--basis transforms can struggle to preserve intricate and highly anisotropic features.

The Barbara\footnote{``Barbara'' image was shared by Professor Allen Gersho from USC. First appearance of
Barbara image in the signal and image processing literature was in the early 1990s.} image is a widely
used benchmark because it combines multiple types of structure in one photograph. Smooth regions of the
face and background test whether a method denoises without introducing artifacts, while sharp contours
along the jawline, fingers, and clothing expose any tendency to blur edges. Most importantly, the
directional textures in the striped scarf and patterned cloth provide a demanding test of a transform's
ability to preserve fine oscillatory detail. These textures are highly anisotropic and rich in
frequency content, and they are precisely the kinds of patterns that many single--basis wavelets
oversmooth. For this reason, Barbara has become a canonical example: a method that performs well on
this image is likely to balance smoothness, edge fidelity, and texture preservation in a realistic and
visually interpretable setting.

To explore this issue, we adopt the same methodology used in earlier experiments (Section~\ref{sec:sec4}). The decomposition
level, type of thresholding, and the universal threshold remain unchanged. Under these identical
conditions, we compare single--basis wavelets with the product transform formed from $\bm W_1$ and
$\bm W_2$. The comparison is instructive: the composite transform, by combining two orthogonal
structures, produces coefficients that adapt more flexibly to local texture. As a result, the
striping in the fabric of Barbara's scarf and trousers is preserved more faithfully. The single--basis transforms
tend to oversmooth these regions, while the product transform maintains both fine oscillatory structure
and global coherence.

We perform a grid search over a library of standard orthonormal wavelet filters (Haar, Daubechies,
Symmlets, and Coiflets) to determine which pairing and ordering $\bm W_1 \circ \bm W_2$ yields the
largest improvement of the product transform over its single--basis components. For each candidate
pair, we estimate the AMSE over $M=200$ Monte Carlo replicates at increasing noise levels $\sigma=20,50,100$,
evaluating performance on both the full Barbara image  (Table~\ref{tab:table3}) and the highly textured scarf region (Table~\ref{tab:table4}). This
search consistently identifies Symm4~$\circ$~Coif3 as the strongest combination, offering the greatest
reduction in MSE and the best preservation of scarf texture relative to either filter alone (Fig.~\ref{fig:fig5}), motivating its use in the remainder of the analysis.

With the same threshold and parallel processing steps, the product transform yields a noticeably
sharper reconstruction of the complex textures. The improvement is visually evident in
Fig.~\ref{fig:fig6} and represents a clear advantage of the composite approach over individual
wavelet bases when preservation of intricate local detail is essential.

\begin{table}[!t]
\centering
\resizebox{\columnwidth}{!}{%
\begin{tabular}{llS[table-format=4.4]S[table-format=3.4]S[table-format=2.2]S[table-format=2.2]S[table-format=1.4]}
\toprule
\textbf{Noise Level} & \textbf{Method} & {\textbf{Average MSE}} & {\textbf{Var.\ of MSE}} & {\textbf{SNR}$_{\mathrm{out}}$ \textbf{(dB)}} & {\textbf{PSNR (dB)}} & {\textbf{SSIM}} \\
\midrule
\multirow{3}{*}{$\sigma = 20$}
& $\bm W_1 \bm W_2$ Product & \bfseries 269.5276 & \bfseries 1.1082 & \bfseries 18.96 & \bfseries 23.83 & \bfseries 0.5003 \\
& $\bm W_1$ (Symm4)   & 338.6568 & 1.6706 & 17.97 & 22.83 & 0.4837 \\
& $\bm W_2$ (Coif3)   & 324.3655 & 1.4576 & 18.16 & 23.02 & 0.4931 \\
\midrule
\multirow{3}{*}{$\sigma = 50$}
& $\bm W_1 \bm W_2$ Product & \bfseries 659.3157 & \bfseries 14.5527 & \bfseries 15.07 & \bfseries 19.94 & \bfseries 0.2967 \\
& $\bm W_1$ (Symm4)   & 915.1187 & 21.4137 & 13.65 & 18.51 & 0.2606 \\
& $\bm W_2$ (Coif3)   & 904.0081 & 24.0717 & 13.70 & 18.57 & 0.2642 \\
\midrule
\multirow{3}{*}{$\sigma = 100$}
& $\bm W_1 \bm W_2$ Product & \bfseries 1091.6725 & \bfseries 72.6674 & \bfseries 12.89 & \bfseries 17.75 & \bfseries 0.1797 \\
& $\bm W_1$ (Symm4)   & 3821.6533 & 996.3215 & 7.44 & 12.31 & 0.1152 \\
& $\bm W_2$ (Coif3)   & 3803.8689 & 997.1222 & 7.46 & 12.32 & 0.1160 \\
\bottomrule
\end{tabular}%
}
\caption{Denoising performance on the full Barbara $512 \times 512$ image for increasing noise
levels $\sigma$. The single--basis transforms $\bm W_1=$ Symm4 and $\bm W_2=$ Coif3 are the classical discrete wavelet transform baseline. The product transform $\bm W_1 \bm W_2$ attains smaller average MSE and variance of MSE, and higher {SNR}$_{\mathrm{out}}$ dB, PSNR, and SSIM than the baseline across the $200$ simulations, with the disparity between the two increasing with noise.}
\label{tab:table3}
\end{table}

\begin{figure}[!t]
\centering
\begin{tabular}{cc}
    \includegraphics[width=0.45\linewidth]{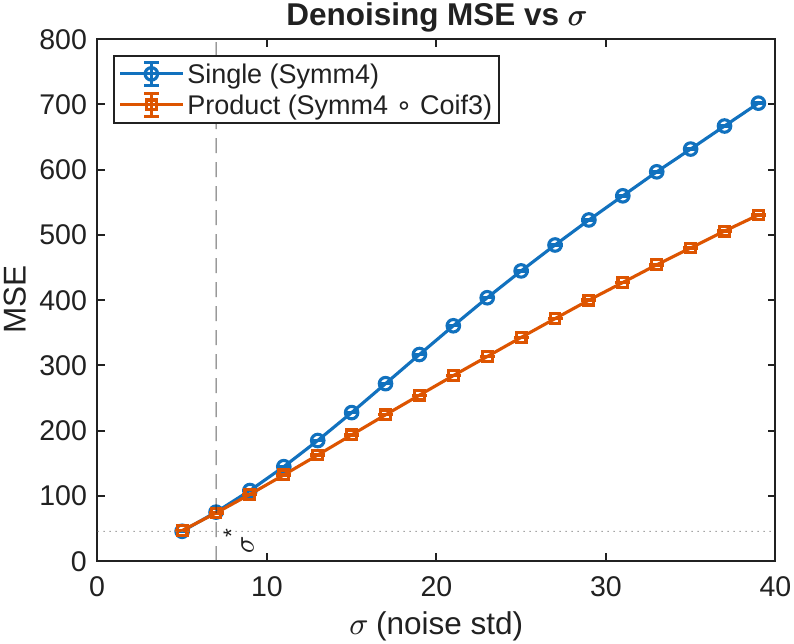} &
    \includegraphics[width=0.45\linewidth]{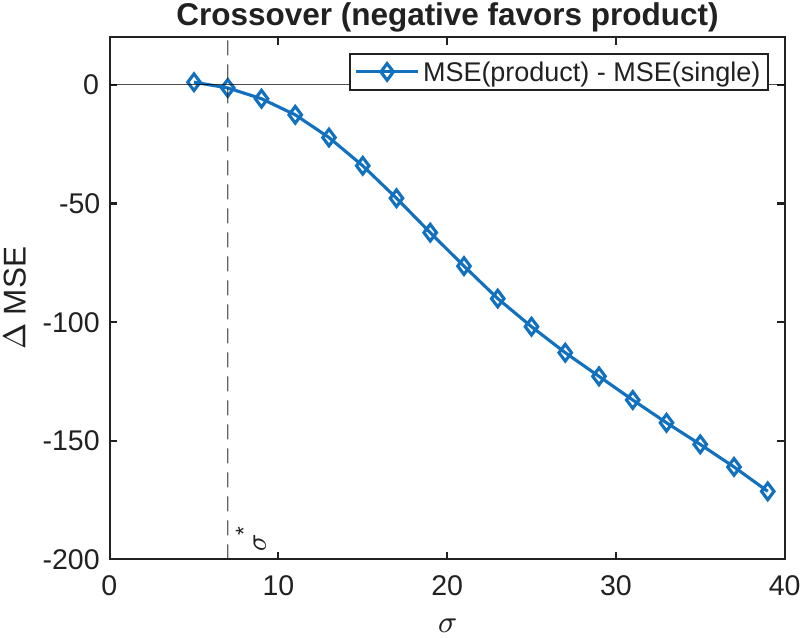} \\
    (a) MSE vs.\ $\sigma$ & (b) $\Delta$MSE vs.\ $\sigma$
\end{tabular}
\caption{Comparison of single--basis transform (Symm4) and product transform (Symm4~$\circ$~Coif3) denoising across
noise levels $\sigma$. Negative $\mathnormal{\Delta}$MSE favors the product transform.}
\label{fig:fig5}
\end{figure}

\begin{table}[!t]
\centering
\resizebox{\columnwidth}{!}{%
\begin{tabular}{lccccc}
\hline
\textbf{Method} & \textbf{Average MSE} & \textbf{Variance of MSE} & {\textbf{SNR}$_{\mathrm{out}}$ \textbf{(dB)}} & \textbf{PSNR (dB)} & \textbf{SSIM} \\
\hline
$\bm W_1 \bm W_2$ Product (Symm4, Coif3)   & \textbf{410.7178} & \textbf{22.3236} & \textbf{16.22} & \textbf{22.00} & \textbf{0.6659} \\
$\bm W_1$ (Symm4)     & 606.0803 & 57.1917 & 14.53 & 20.31 & 0.5859 \\
$\bm W_2$ (Coif3)     & 571.1162 & 45.5754 & 14.78 & 20.56 & 0.6017 \\
\hline
\end{tabular}%
}
\caption{Denoising performance on the Barbara scarf patch across product and
single--basis transforms, averaged over $200$ simulations ($\sigma = 20$). The product transform $\bm W_1\bm W_2$ gives the
lowest MSE and the highest {SNR}$_{\mathrm{out}}$ dB, PSNR, and SSIM, showing that it preserves more
of the scarf texture than either single--basis transform under the identical thresholding
rule.}
\label{tab:table4}
\end{table}

\begin{figure}[!t]
\centering
\begin{tabular}{cc}
    \includegraphics[width=0.25\linewidth]{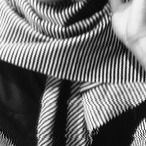} &
    \includegraphics[width=0.25\linewidth]{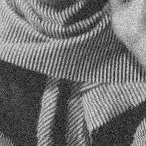} \\[0.7em]
  {\footnotesize  (a) Original scarf patch} &  {\footnotesize (b) Noisy scarf patch} \\[1.2em]
    \includegraphics[width=0.25\linewidth]{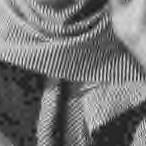} &
    \includegraphics[width=0.25\linewidth]{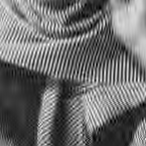} \\[0.7em]
  {\footnotesize  (c) Denoised (Single-basis transform) } & {\footnotesize (d) Denoised (Product transform $\bm W_1 \bm W_2$)}
\end{tabular}
\caption{Comparison of original, noisy, and denoised scarf region of the Barbara image ($\sigma = 20$). Denoising is
performed using hard thresholding with universal threshold. Note that the stripes are more defined and
clearer in (d), with the bottom edge of the scarf being more distinguishable than in (c).}
\label{fig:fig6}
\end{figure}

\subsection{Energy Concentration in Turbulence Signals and Townsend--Type Partitioning}
\label{sec:sec5b}

Time series measurements of wind velocity and temperature $T$ were collected over a grass--covered
forest clearing at Duke Forest near Durham, North Carolina. The components of wind speed were measured
using a Gill triaxial sonic anemometer at the sampling frequency $f_s = 56\ \mbox{Hz}$.

Here we consider the vertical velocity component of length $N=2048$ and show that product wavelet
transforms compress data better than individual orthogonal wavelet transforms (see Table~\ref{tab:table5}, Figure~\ref{fig:fig7} and Appendix~\ref{sec:B}). Atmospheric turbulent
flows are highly intermittent and, under the Kolmogorov K41 law, have Hurst exponent $1/3$. This
suggests that more complex atomic functions in the composite framelet decomposition can be more
economical in representing signal energy in such time series.

\begin{table}[!t]
\centering
\begin{tabular}{c|ccc}
\empty & Total energy & Top 1\% coefs & Top 5\% coefs \\
\hline
$\bm W_1$          & 1146.6438 & 0.2374 & 0.6419 \\
$\bm W_2$          & 1146.6438 & 0.2410 & 0.6386 \\
$\bm W_1 \bm W_2$  & 1146.6438 & \textbf{0.4558} & \textbf{0.8372} \\
$\bm W_2 \bm W_1$  & 1146.6438 & \textbf{0.4630} & \textbf{0.8338}
\end{tabular}
\caption{Both product transforms yield sparser representations by concentrating signal energy into fewer coefficients.}
\label{tab:table5}
\end{table}

\begin{figure}[!t]
\centering
\includegraphics[width=0.70\linewidth]{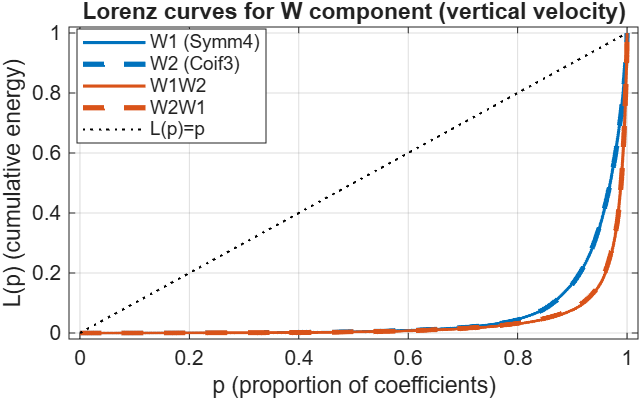}
\caption{The two single--basis transforms of Symmlet~4 and Coiflet~3 produce nearly identical Lorenz curves, as do
the two product transforms $\bm W_1 \bm W_2$ and $\bm W_2 \bm W_1$. The product transforms clearly dominate
the single--basis transforms in terms of energy concentration.}
\label{fig:fig7}
\end{figure}

As observed in \citet{KatulVidakovic1996}, wavelet--based representations, and especially those that
improve sparsity relative to classical orthogonal wavelets, provide sharper discrimination between
coherent eddy signatures and background fluctuations. Enhanced compressibility therefore strengthens
the efficacy of Townsend--style~\citep{Townsend1976} decompositions by allowing the attached--eddy
contribution to be isolated with fewer and more interpretable coefficients, while noise--like detached
structures are more effectively suppressed through thresholding.

\section{Discussion}
\label{sec:VI}
From an engineering standpoint, the central message of this work is that denoising performance can be substantially
improved without sacrificing the properties that make wavelets attractive in practice.
All composite constructions studied here preserve orthogonality (or unitarity), and therefore retain
energy preservation, numerical stability, and perfect reconstruction. This ensures that the standard
transform–shrink–inverse pipeline remains valid, predictable, and easy to deploy, even when the
transform no longer corresponds to a classical two–channel wavelet filterbank.

The empirical results show that relaxing the strict requirement of a single filterbank representation
opens a richer design space for multiscale transforms. Composite wavelet--like matrices, especially
products and Kronecker products, often induce stronger concentration of signal energy into a
smaller set of coefficients. This effect is clearly reflected in the Lorenz curve diagnostics and
translates directly into improved denoising under identical thresholding rules. Across benchmark
signals, real data, and image examples, the product transform $\bm W_1 \bm W_2$ repeatedly outperforms
single--basis wavelets in terms of MSE, while requiring no additional reconstruction machinery.

An important practical lesson is that orthogonality alone is not sufficient to guarantee good
denoising performance. Similarity transforms preserve energy and invertibility, yet their tendency to
mix scales and locations weakens sparsity and reduces effectiveness under thresholding. In contrast,
products and Kronecker constructions maintain a meaningful multiscale organization of coefficients,
which is essential for separating signal from noise. Block--diagonal composites further demonstrate
that adaptivity to heterogeneous signal structure can be introduced in a controlled and interpretable
way, again without sacrificing orthogonality.

Taken together, these findings suggest a shift in perspective that is useful for engineering
applications. Rather than viewing wavelet transforms solely through the lens of classical filterbank
theory, it is advantageous to regard them as orthogonal operators that can be combined algebraically
to better match signal structure. Lorenz curves and MSE summaries provide simple, interpretable tools
for assessing whether a given composite transform improves sparsity and denoising performance.

Looking ahead, several directions are particularly relevant for practice. More advanced shrinkage
rules, including block, Bayesian, or empirical Bayes methods, are likely to benefit even more from
the heavier--tailed coefficient distributions produced by composite transforms. On the computational
side, efficient implementations of products and Kronecker constructions, especially on GPUs or
tensor--based architectures, merit further study. These algebraic constructions could also be connected with recent learning-based approaches, such as wavelet-inspired invertible networks~\citep{HuangDragotti2022}, which move beyond a single fixed basis through learned nonlinear transforms; combining their adaptivity with the interpretability of linear orthogonal operators is a natural avenue for future work. Finally, the unitary structure of these composites
naturally aligns with tensor and quantum computing frameworks, suggesting broader applicability
beyond classical signal denoising. 

In summary, composite wavelet--like matrices preserve the stability guarantees engineers rely on,
while offering increased flexibility, stronger sparsity, and improved denoising performance. They
constitute a practical and effective extension of classical wavelet methods for modern signal and
data analysis tasks.

\section*{Acknowledgment}

Work of Radhika Kulkarni was supported by the H. O. Hartley Chair Foundation at Texas A\&M University. Vidakovic acknowledges NSF Award 2515246 at the same institution. The authors thank the Editor and three anonymous referees, whose comments greatly improved the presentation.

\appendix

\section{Product Transforms: Proof of Non-Wavelet Structure and the Induced Atoms}
\label{sec:A}
\label{sec:appendix}

This appendix explains why the product
\begin{eqnarray}
\bm W &=& \bm W_1 \bm W_2
\label{eq:app-product-transform}
\end{eqnarray}
of two orthogonal wavelet matrices is, in general, not itself the matrix of a single classical two-channel wavelet filterbank, even though it is orthogonal and perfectly reconstructing.

Let $\bm W_1$ and $\bm W_2$ be orthogonal wavelet transform matrices of size $N \times N$, generated by compactly supported orthonormal two-channel filterbanks. Denote their low-pass analysis filters by $h_1$ and $h_2$, and their high-pass filters by $g_1$ and $g_2$. Then
\begin{eqnarray}
\bm W^\top \bm W
&=&
(\bm W_1 \bm W_2)^\top (\bm W_1 \bm W_2) \nonumber\\
&=&
\bm W_2^\top \bm W_1^\top \bm W_1 \bm W_2 \nonumber\\
&=&
\bm W_2^\top \bm W_2 \nonumber\\
&=&
I.
\label{eq:app-product-orthogonality}
\end{eqnarray}
Thus $\bm W$ is an orthogonal transform and the inverse transform is $\bm W^\top$. The separate question is whether $\bm W$ can be represented as a single classical two-channel discrete wavelet transform generated by some analysis filters $(h_A,g_A)$.

A natural attempt to collapse the product into a single two-channel filterbank is to combine the two low-pass filters as
\begin{eqnarray}
h_A &=& h_1 * h_2^{\uparrow 2},
\label{eq:app-collapsed-lowpass}
\end{eqnarray}
where $h_2^{\uparrow 2}$ denotes the filter obtained by inserting one zero between consecutive taps of $h_2$. If $H_1(z)$ and $H_2(z)$ are the corresponding $z$-transforms, then the proposed collapsed low-pass filter has
\begin{eqnarray}
H_A(z) &=& H_1(z) H_2(z^2).
\label{eq:app-HA-z}
\end{eqnarray}

For an orthonormal two-channel filterbank, the low-pass filter must satisfy the half-band, or quadrature mirror, condition
\begin{eqnarray}
\Delta_A(\omega) &=& C_A,
\label{eq:app-halfband-condition}
\end{eqnarray}
where $C_A$ is a constant independent of $\omega$, and
\begin{eqnarray}
\Delta_A(\omega)
&=&
|H_A(e^{i\omega})|^2
+
|H_A(-e^{i\omega})|^2 .
\label{eq:app-DeltaA}
\end{eqnarray}
The precise value of $C_A$ depends on the normalization convention; the essential requirement is that $\Delta_A(\omega)$ be constant in $\omega$.

For the collapsed filter in \eqref{eq:app-HA-z},
\begin{eqnarray}
H_A(e^{i\omega})
&=&
H_1(e^{i\omega}) H_2(e^{2i\omega}),
\label{eq:app-HA-omega}
\end{eqnarray}
and
\begin{eqnarray}
H_A(-e^{i\omega})
&=&
H_A(e^{i(\omega+\pi)}) \nonumber\\
&=&
H_1(e^{i(\omega+\pi)})
     H_2(e^{2i(\omega+\pi)}) \nonumber\\
&=&
H_1(-e^{i\omega}) H_2(e^{2i\omega}),
\label{eq:app-HA-minus-omega}
\end{eqnarray}
because $e^{2i(\omega+\pi)}=e^{2i\omega}$. Substituting \eqref{eq:app-HA-omega} and \eqref{eq:app-HA-minus-omega} into \eqref{eq:app-DeltaA} gives
\begin{eqnarray}
\Delta_A(\omega)
&=&
|H_2(e^{2i\omega})|^2
\left[
|H_1(e^{i\omega})|^2
+
|H_1(-e^{i\omega})|^2
\right].
\label{eq:app-DeltaA-expanded}
\end{eqnarray}
Since $h_1$ is an orthonormal low-pass filter,
\begin{eqnarray}
|H_1(e^{i\omega})|^2
+
|H_1(-e^{i\omega})|^2
&=&
C_1,
\label{eq:app-H1-halfband}
\end{eqnarray}
where $C_1$ is constant in $\omega$. Therefore
\begin{eqnarray}
\Delta_A(\omega)
&=&
C_1 |H_2(e^{2i\omega})|^2 .
\label{eq:app-DeltaA-final}
\end{eqnarray}
For a nontrivial compactly supported orthonormal wavelet low-pass filter $h_2$, the quantity $|H_2(e^{2i\omega})|^2$ is not constant in $\omega$. A genuine low-pass wavelet filter has frequency selectivity, with large response near zero frequency and attenuation near the Nyquist frequency. Hence the candidate collapsed filter $h_A$ does not satisfy the half-band condition required of a single orthonormal two-channel wavelet filterbank, except in degenerate cases.

Even though $\bm W$ is not a classical wavelet matrix, it defines a valid orthonormal system.
The decomposing atoms (analysis functions) are obtained directly from the matrix:
for the $k$th coefficient, the associated atom is

\[
\xi_k = \bm W^\top e_k,
\]
and for any discrete signal $f\in\mathbb{R}^N$ one has the exact expansion

\[
f = \sum_{k=1}^N \langle f,\xi_k\rangle\,\xi_k .
\]
Thus the atoms are precisely the columns of $\bm W^\top$ (equivalently, the rows of $\bm W$).
In this sense, $\bm W$ generates a tight frame, in fact an orthonormal basis, but the atoms
do not arise from a single refinable scaling function or a two--channel refinement/cascade
construction. Classical refinement-based machinery presupposes a perfect reconstruction
polyphase condition; the failure of the half-band condition shown above rules that out for the composite
filter attempt, so the atoms must be read from the matrix itself.

Figure~\ref{fig:fig8} illustrates this concretely for a 
representative index $k=1000$ when both $\bm W_1$
and $\bm W_2$ are $2048\times 2048$ transforms (Coiflet 1 for $\bm W_1$ and Symmlet 4 for
$\bm W_2$). Panel (a) shows the full atom $\bm W^\top(1000,:)$,
and Panel (b) shows a zoomed segment revealing localized oscillations. The key engineering
message is that product atoms are \emph{mixtures} of two wavelet systems: they remain
energy-preserving and perfectly invertible, but their shapes are richer than classical wavelet
translates/dilates. This added flexibility is exactly what can improve sparsity and denoising
after a fixed threshold: the dictionary has more varied atom morphologies, so typical signals
can align more closely with a smaller subset of coefficients.

\begin{figure}[!t]
\centering
\begin{minipage}[t]{0.49\columnwidth}
    \centering    \includegraphics[width=\textwidth]{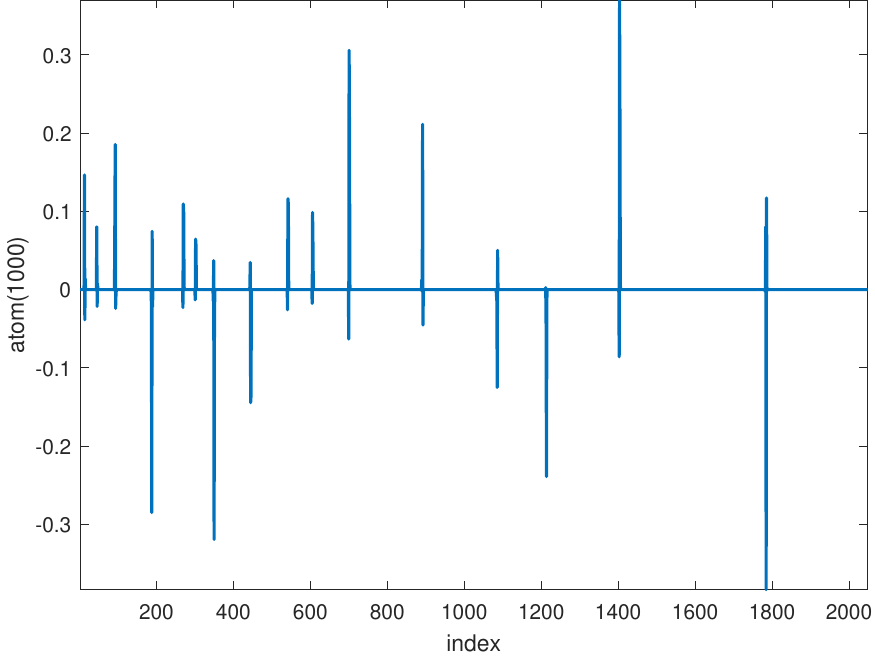}
\end{minipage}
\hfill
\begin{minipage}[t]{0.49\columnwidth}
    \centering
    \includegraphics[width=\textwidth]{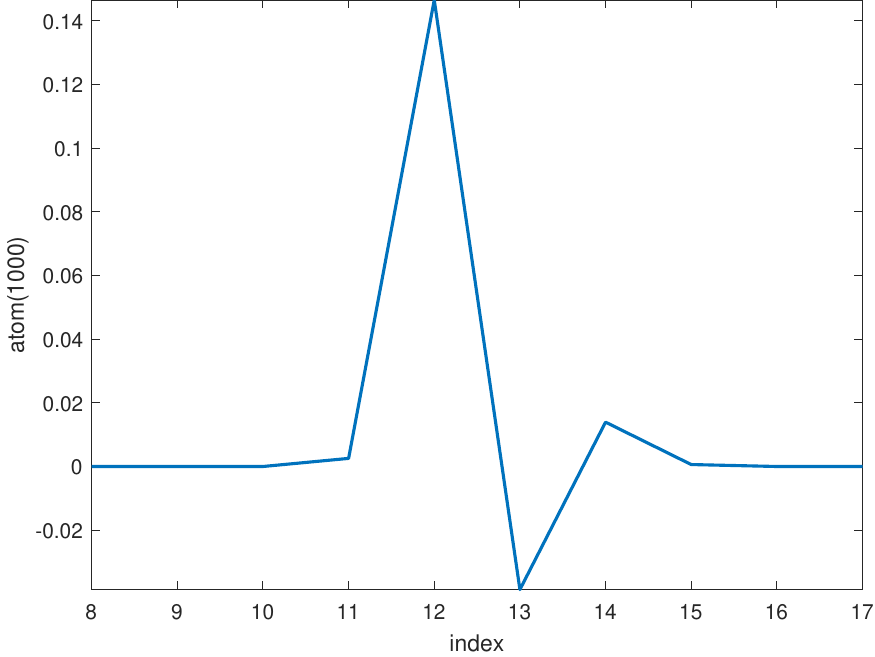}
\end{minipage}
\caption{\small (a) Composite framelet atom extracted from position $k=1000$ of $ \bm W$, generated by the composite transform $ \bm W = \bm W_{1} \bm W_{2}$.
(b) Zoomed-in region illustrating local behavior of the same atom. }
\label{fig:fig8}
\end{figure}
In summary, the product transform $\bm W_1\bm W_2$ should be viewed as an orthogonal
composite transform whose atoms are best interpreted as framelet-like basis elements defined
by the matrix product itself. They preserve perfect reconstruction by orthogonality, while
lying outside the classical two--channel wavelet filter bank class.

\section{Additional Empirical Evidence: Turbulence Components}
\label{sec:B}

We extend the discussion of turbulence signal partitioning in Section~\ref{sec:sec5b} by examining not only the effect of the product composite matrices on the vertical velocity component (W), but also on the remaining data streams.  Figure~\ref{fig:fig9} and Table~\ref{tab:energy_concentration} present the results of applying $\bm W_1 \bm W_2$ and $\bm W_2 \bm W_1$ to temperature (T) and to the streamwise (U), lateral (V) and vertical (W) velocity components both in graphical and tabular form. 

\begin{figure}[!t]
\centering

\begin{tabular}{cc}
    \includegraphics[width=0.46\linewidth]{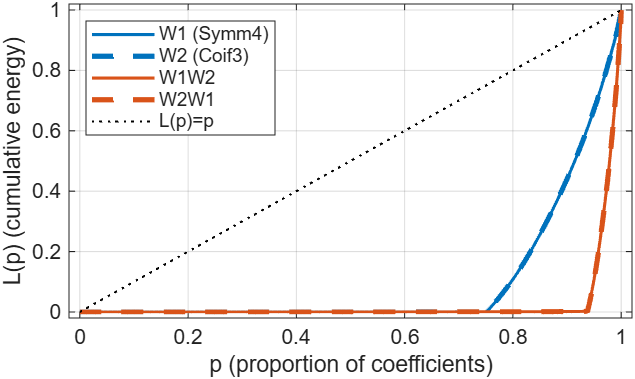} &
    \includegraphics[width=0.46\linewidth]{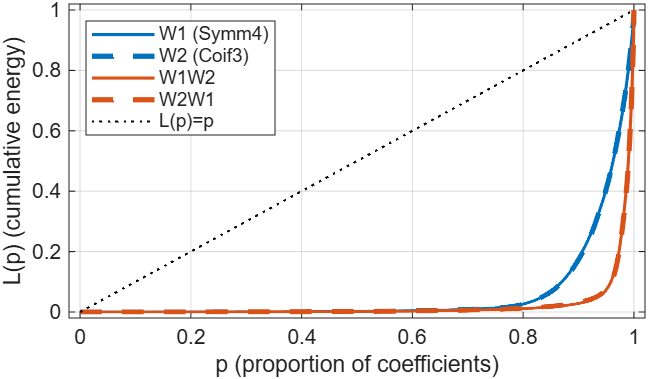} \\[0.4em]
    (a) Stream Velocity (U) & (b)  Lateral Velocity (V) \\[0.45em]

    \includegraphics[width=0.46\linewidth]{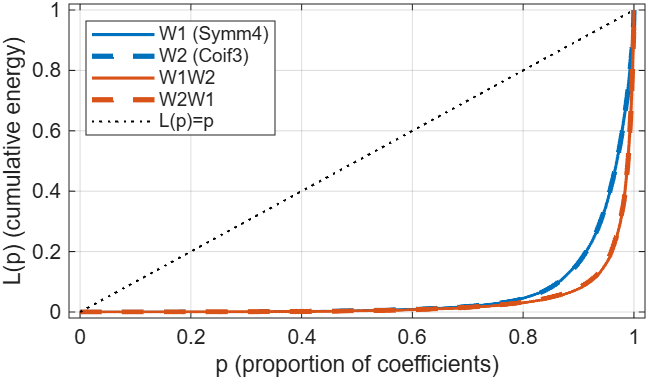} &
    \includegraphics[width=0.46\linewidth]{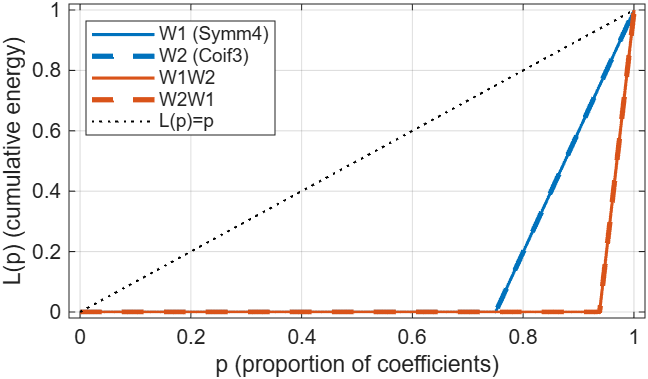} \\[0.4em]
    (c) Vertical Velocity (W) & (d)  Temperature (T)
\end{tabular}
\caption{\textbf{Lorenz curves for U, V, W, and T components of Turbulence data.} Both product transforms generate far greater disbalance of the energies of the data streams in comparison to the single--basis transforms. The same filters as Barbara example in Section~\ref{sec:sec5A} are used for continuity.}
\label{fig:fig9}
\end{figure}

\medskip
\begin{table}[!t]
\centering

\begin{tabular}{ll| S[table-format=9.4] S[table-format=1.4] S[table-format=1.4]}
\toprule
\textbf{Component} & \textbf{Transform} & \textbf{Total Energy} & \textbf{Top 1\%} & \textbf{Top 5\%} \\
\midrule
U (streamwise) 
& $\bm W_1$  & 65749.1668 & 0.0807 & 0.3251 \\
& $\bm W_2$ & 65749.1668 & 0.0805 & 0.3252 \\
& $\bm W_1 \bm W_2$      & 65749.1668 & \textbf{0.2634} & \textbf{0.8840} \\
& $\bm W_2 \bm W_1$       & 65749.1668 & \textbf{0.2635} & \textbf{0.8852} \\
\midrule
V (lateral) 
& $\bm W_1$ & 3288.3322 & 0.1867 & 0.5889 \\
& $\bm W_2$ & 3288.3322 & 0.1867 & 0.5880 \\
& $\bm W_1 \bm W_2$       & 3288.3322 & \textbf{0.4782} & \textbf{0.9248} \\
& $\bm W_2 \bm W_1$       & 3288.3322 & \textbf{0.4843} & \textbf{0.9251} \\
\midrule
W (vertical) 
& $\bm W_1$ & 1146.6438 & 0.2374 & 0.6419 \\
& $\bm W_2$ & 1146.6438 & 0.2410 & 0.6386 \\
& $\bm W_1 \bm W_2$       & 1146.6438 & \textbf{0.4558} & \textbf{0.8372} \\
& $\bm W_2 \bm W_1$       & 1146.6438 & \textbf{0.4630} & \textbf{0.8338} \\
\midrule
T (temperature) 
& $\bm W_1$ & 189165425.5900 & 0.0391 & 0.1995 \\
& $\bm W_2$ & 189165425.5914 & 0.0391 & 0.1995 \\
& $\bm W_1 \bm W_2$       & 189165425.5915 & \textbf{0.1575} & \textbf{0.7982} \\
& $\bm W_2 \bm W_1$       & 189165425.5915 & \textbf{0.1570} & \textbf{0.7977} \\
\bottomrule
\end{tabular}

\caption{Energy concentration of wavelet and product transforms ($\bm W_1$ Symmlet 4, $\bm W_2$ Coiflet 3) for the four signal components ($N=2048$ samples). Reported values show total energy and the fraction of energy captured by the top 1\% and top 5\% largest coefficients. }
\label{tab:energy_concentration}
\end{table}

The results indicate that both product transforms $\bm W_1 \bm W_2$ and $\bm W_2 \bm W_1$ induce a substantially greater degree of energy disbalance, particularly for the streamwise velocity and temperature signals. As expected, the total energy is preserved under the composite transforms due to orthogonality.

\section{Proof of Corollary 1 (Universal Threshold Validity)}
\label{sec:C}

\begin{proof}
By Lemma~\ref{lem:noise}, the coordinates
$z_i := (\bm W\uw\varepsilon)_i$ are i.i.d.\ $\mathcal{N}(0,\sigma^2)$.
Let $\varphi$ denote the standard normal density. For $t>0$ and
$Z\sim\mathcal{N}(0,1)$, since $u/t \ge 1$ on $[t,\infty)$ and
$\varphi'(u) = -u\varphi(u)$,
\[
\Pr(Z > t) = \int_t^{\infty}\!\varphi(u)\,du
\;\le\; \int_t^{\infty}\!\frac{u}{t}\,\varphi(u)\,du
= \frac{\varphi(t)}{t}.
\]
Hence for $z_i\sim\mathcal{N}(0,\sigma^2)$ (and $|z_i|$ two-sided) ,
\[
\Pr(|z_i| > t)
\;\le\; \frac{2\,\varphi(t/\sigma)}{t/\sigma}
= \sqrt{\tfrac{2}{\pi}}\;\frac{\sigma}{t}\;e^{-t^2/2\sigma^2}.
\]
Taking union bound over the $N$ coordinates gives us 
\[
\Pr\!\Big(\max_{1\le i\le N}|z_i| > t\Big)
\;\le\; N\sqrt{\tfrac{2}{\pi}}\;\frac{\sigma}{t}\;e^{-t^2/2\sigma^2}.
\]
Now we substitute the universal threshold  $t = \lambda = \sigma\sqrt{2\log N}$, which
gives $e^{-t^2/2\sigma^2} = e^{-\log N} = 1/N$ and
$\sigma/t = 1/\sqrt{2\log N}$, so
\[
\Pr\!\Big(\max_{1\le i\le N}|z_i| > \lambda\Big)
\;\le\; \sqrt{\tfrac{2}{\pi}}\;\frac{1}{\sqrt{2\log N}}
= \frac{1}{\sqrt{\pi\log N}} \;\xrightarrow[N\to\infty]{}\; 0 .
\]
Note that the bound is non-asymptotic: when $N=1024$ it equals
$1/\sqrt{\pi\log 1024}\approx 0.21$. This is a conservative upper bound on the probability that \emph{any} noise coefficient exceeds the threshold; the resulting contribution to the risk is controlled by Corollary~\ref{cor:oracle}.
\end{proof}

\section{Robustness Across Input-SNR Levels and Further Comparative Simulation}
\label{sec:D}

This appendix reports the one-dimensional robustness study referenced in
Section~\ref{sec:IV}. Using the same simulation pipeline (length $N=1024$, $M=200$ replications, universal threshold, hard
thresholding), we consider
 $\mathrm{SNR}_{\mathrm{in}}\in\{-5,0,5,8.45\}$~dB under the power definition
~\eqref{eq:snr_db}.  The
range spans strongly noise-dominated conditions to a comparatively clean regime
and includes the default level $\mathrm{SNR}_{\mathrm{in}}=5$~dB. For each Donoho--Johnstone signal, the product transform
$\bm W_1\bm W_2$ is compared with $\bm W_1$ under the same thresholding rule.

Fig.~\ref{fig:fig10} and Table~\ref{tab:snr_sweep} show lower post-denoising AMSE for the
product at every tested level and SNR, and the ratio is largest at low SNR, consistent with the oracle-risk analysis of
Section~\ref{sec:theory}. 

\begin{figure}[!t]
\centering
\includegraphics[width=0.70\linewidth]{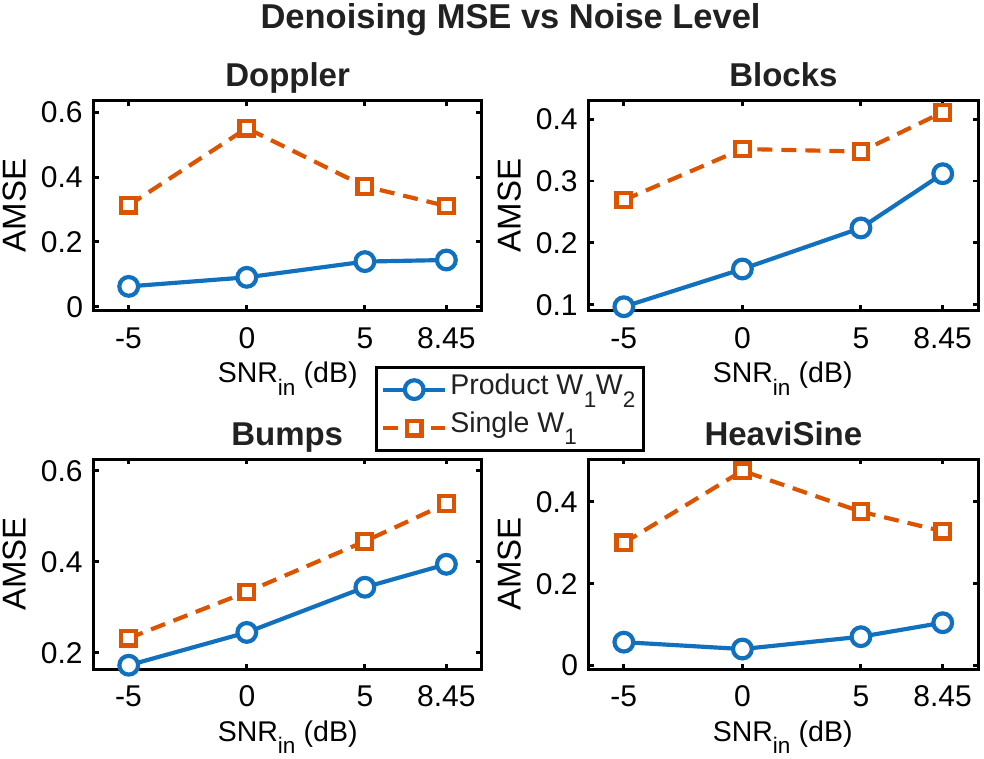}
\caption{Post-denoising AMSE versus SNR (in dB) for the four
Donoho--Johnstone signals ($N=1024$, $200$ replications). The product transform
$\bm W_1\bm W_2$ attains lower AMSE than the single--basis transform $\bm W_1$ at every
tested  input-SNR level.}
\label{fig:fig10}
\end{figure}

\begin{table}[!t]
\centering
\caption{ Post-denoising AMSE ($200$ replications, $N=1024$) for the product
transform $\bm W_1\bm W_2$ and its single--basis transform $\bm W_1$ across noise levels.
$\mathrm{SNR}_{\mathrm{in}}$ is the power ratio~\eqref{eq:snr_db}, in
decibels. The ratio row is the single-basis AMSE divided by the product AMSE;
values above one favor the product transform. For each signal the lower AMSE is
in bold.}
\label{tab:snr_sweep}
\begin{tabular}{ll S[table-format=1.4] S[table-format=1.4] S[table-format=1.4] S[table-format=1.4]}
\toprule
\textbf{Signal} & \textbf{Transform}
  & {$-5$~dB} & {$0$~dB} & {$5$~dB} & {$8.45$~dB} \\
\midrule
\multirow{3}{*}{Doppler}
 & Product $\bm W_1\bm W_2$ & \bfseries 0.0628 & \bfseries 0.0907 & \bfseries 0.1393 & \bfseries 0.1442 \\
 & Single basis $\bm W_1$   & 0.3132 & 0.5509 & 0.3715 & 0.3108 \\
 & Ratio                    & {4.99} & {6.07} & {2.67} & {2.16} \\
\midrule
\multirow{3}{*}{Blocks}
 & Product $\bm W_1\bm W_2$ & \bfseries 0.0967 & \bfseries 0.1577 & \bfseries 0.2240 & \bfseries 0.3117 \\
 & Single basis $\bm W_1$   & 0.2691 & 0.3516 & 0.3475 & 0.4110 \\
 & Ratio                    & {2.78} & {2.23} & {1.55} & {1.32} \\
\midrule
\multirow{3}{*}{Bumps}
 & Product $\bm W_1\bm W_2$ & \bfseries 0.1728 & \bfseries 0.2449 & \bfseries 0.3441 & \bfseries 0.3946 \\
 & Single basis $\bm W_1$   & 0.2322 & 0.3338 & 0.4451 & 0.5278 \\
 & Ratio                    & {1.34} & {1.36} & {1.29} & {1.34} \\
\midrule
\multirow{3}{*}{HeaviSine}
 & Product $\bm W_1\bm W_2$ & \bfseries 0.0565 & \bfseries 0.0396 & \bfseries 0.0699 & \bfseries 0.1042 \\
 & Single basis $\bm W_1$   & 0.2997 & 0.4763 & 0.3761 & 0.3283 \\
 & Ratio                    & {5.30} & {12.01} & {5.38} & {3.15} \\
\bottomrule
\end{tabular}
\end{table}

\begin{figure}[!t]
\centering
\includegraphics[width=0.90\linewidth]{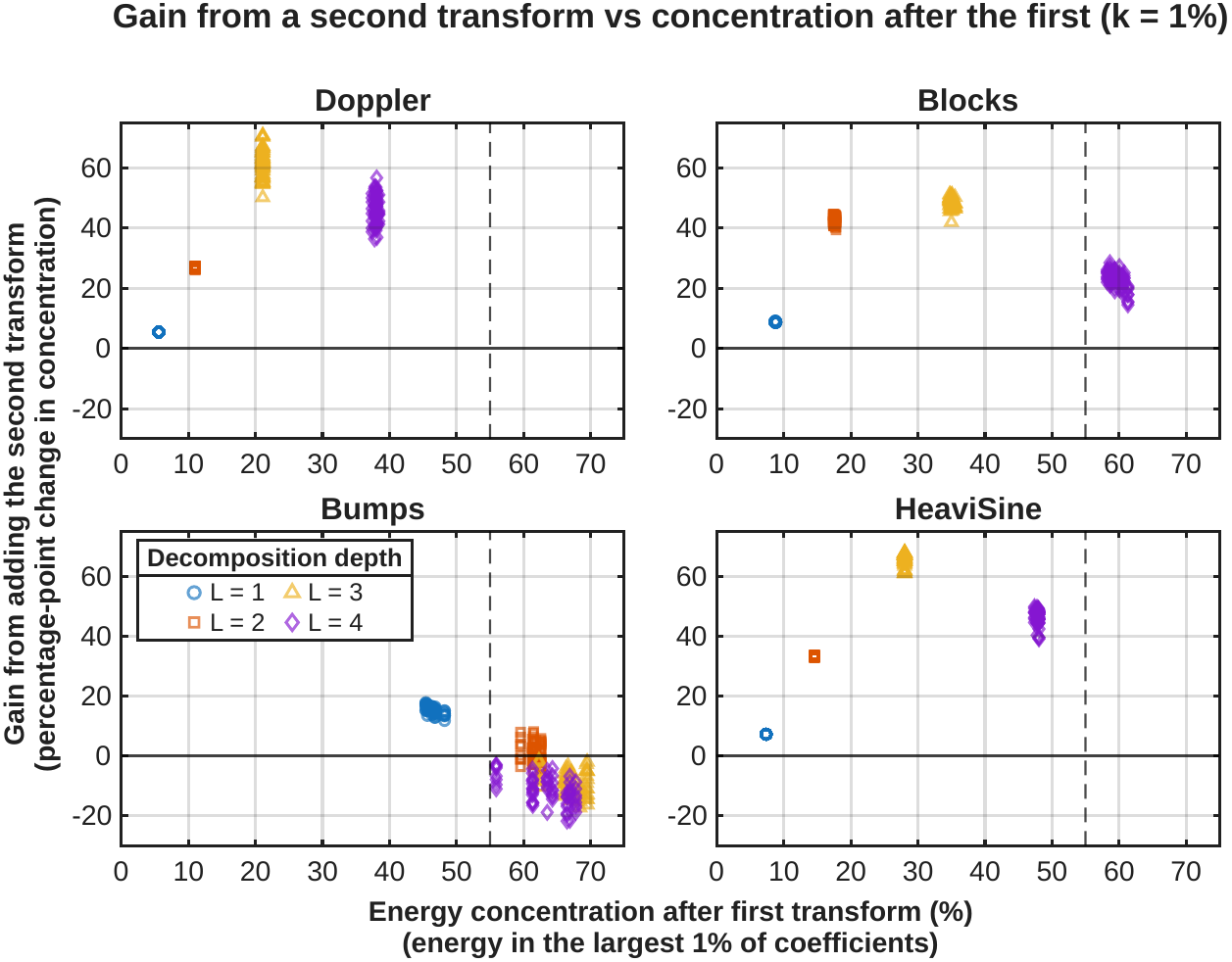}
\caption{ \textbf{When the product transform improves energy concentration.}
Gain in top-$1\%$ energy concentration from adding a second transform, $C_k(\bm W_1\bm W_2\uw s)-C_k(\bm W_2\uw s)$, versus the concentration achieved by the first transform alone, $C_k(\bm W_2\uw s)$. Each point represents one of $1024=64\times4\times4$ configurations,
combining an ordered filter pair, a decomposition depth $L=1,\ldots,4$,
and one of the four Donoho--Johnstone signals. The horizontal line marks zero gain, and the vertical dashed
line marks $\tau=0.55$, the largest inner concentration below which every configuration improves.}
\label{fig:fig11}
\end{figure}

Fig.~\ref{fig:fig11}  reports the exhaustive filter-pair sweep across a library of Haar, Daubechies 2--5, Symmlet 4 and 5, and Coiflet 1. Across all four signals, lower concentration after the inner transform $\bm W_2$ is associated with larger gains from the product transform.
At the top-$1\%$ level, every tested configuration with
$C_k(\bm W_2\uw s)<0.55$ improves, while above this threshold improvement
is no longer guaranteed. The same qualitative pattern holds at the top-$5\%$ and $10\%$ levels, with corresponding thresholds $\tau=0.93$ and $\tau=0.95$. Thus, the product is most beneficial when the inner transform leaves the signal energy relatively spread out; full numerical results are provided
in the supplementary materials on GitHub.

\section{Non-Universality of Sparsity Improvement}
\label{sec:E}

\begin{proposition}
\label{prop:nonuniversality}
Let $\bm W_1$ and $\bm W_2$ be $N\times N$ orthogonal matrices. If
\(
C_1(\bm W_1\bm W_2 y)\ \geq\ C_1(\bm W_2 y)
\text{ for every nonzero } y\in\mathbb{R}^N,
\)
then $\bm W_1$ is a signed permutation matrix, and consequently
\[
C_k(\bm W_1\bm W_2 y) = C_k(\bm W_2 y),\qquad k=1,\ldots,N,
\]
for every nonzero $y$.
\end{proposition}

\begin{proof}
Since $\bm W_2$ is onto, set $x=\bm W_2 y$ and take $x=e_j$. The assumption gives
$C_1(\bm W_1 e_j)=1$, so each column of $\bm W_1$ has exactly one nonzero entry.
Orthogonality forces these entries to have magnitude one and to occur in
distinct rows. Hence $\bm W_1$ is a signed permutation matrix, which preserves
all ordered coefficient magnitudes and therefore every $C_k$.
\end{proof}

Thus, a nontrivial product transform cannot uniformly improve concentration
for all signals, including over both of its component transforms. Any strict gain depends on the signal, the order of
composition, and the selected decomposition levels. For unitary matrices,
the corresponding conclusion is that $\bm W_1$ must be a monomial unitary
matrix.

\noindent Fig.~\ref{fig:fig11} in Appendix~\ref{sec:D} reports the corresponding empirical study of the cases in which product improves concentration, and shows that improvement is strongly related to the
concentration already achieved by the inner transform.
\newpage
\bibliography{references}

\end{document}